\documentclass[3p,times, 11pt]{elsarticle}
 \usepackage[latin1]{inputenc}
\usepackage{amssymb,amsmath,amsthm,amscd,latexsym}
\usepackage{mathrsfs}
\usepackage{amsfonts}
\usepackage{amscd}
\usepackage{array}
\usepackage{hyperref}
\hypersetup{
    colorlinks=true,
    linkcolor=blue,
    citecolor=blue,
    urlcolor=blue
}
\usepackage{times}
\usepackage{listings}
\usepackage{mathtools}
\usepackage{mathdots}

\newtheorem{Definition}{Definition}
\newtheorem{Theorem}{Theorem}

\newtheorem{Example}{Example}[section]
\newtheorem{Remark}{Remark}

\newtheorem{Definition and Notation}{Definition and Notation}
\newtheorem{Lemma}{Lemma}

\newtheorem{Proposition}{Proposition}
\newtheorem{Corollary}{Corollary}

\journal{ Cryptography and Communications}
\begin{document}

\begin{frontmatter}

\title{Polycyclic codes over serial rings and their annihilator CSS construction}

\author[1,m,v]{Maryam Bajalan}  
\ead{maryam.bajalan@math.bas.bg}
\fntext[m]{Maryam Bajalan is supported by the Bulgarian Ministry of Education and Science, Scientific Programme ``Enhancing the Research Capacity in Mathematical Sciences (PIKOM)", No. DO1-67/05.05.2022.}

\author[2,v]{Edgar Mart\'inez-Moro} 
\ead{edgar.martinez@uva.es}
\fntext[v]{Both authors were partially supported by   Grant TED2021-130358B-I00 funded by MCIU/AEI/10.13039/501100011033 and by the ``European Union NextGenerationEU/PRTR". This research was partially conducted during a research stay of the second author funded by the Institute of Mathematics and Informatics at the Bulgarian Academy of Sciences.}

\address[1]{Institute of Mathematics and Informatics\\ Bulgarian Academy of Sciences\\Bl. 8, Acad. G. Bonchev Str., 1113, Sofia, Bulgaria}
\address[2]{Institute of Mathematics\\ University of Valladolid, Castilla, Spain}

\begin{abstract}
In this paper, we investigate the algebraic structure for polycyclic codes over a specific class of serial rings, defined as $\mathscr R=R[x_1,\ldots, x_s]/\langle t_1(x_1),\ldots, t_s(x_s) \rangle$,  where $R$ is a chain ring and each $t_i(x_i)$ in $R[x_i]$ for $i\in\{1,\ldots, s\}$ is a monic square-free polynomial. We define quasi-$s$-dimensional polycyclic codes and establish an $R$-isomorphism between these codes and polycyclic codes over $\mathscr R$. We provide necessary and sufficient conditions for the existence of annihilator self-dual, annihilator self-orthogonal, annihilator linear complementary dual, and annihilator dual-containing polycyclic codes over this class of rings. We also establish the CSS construction for annihilator dual-preserving polycyclic codes over Frobenius rings.

\end{abstract}
 

\end{frontmatter}
\section{Introduction}
Polycyclic codes are ideals in the quotient of the polynomial ring $A[x]/\langle f(x)\rangle$, where $A$ is a ring and $f(x)$ is a polynomial in $A[x]$, see \cite{Alahmadi_2016,Bajalan_2022, Bajalan_2024, Permouth_2009}. Cyclic codes and constacyclic codes over $A$ are particular instances, when $f(x)=x^n-1$ and $f(x)=x-\lambda$, respectively. A ring is called uniserial if all its submodules are linearly ordered by inclusion, and it is called serial if it is a direct sum of uniserial rings. In this paper, we aim to investigate polycyclic codes in the case $A$ is the serial ring $\mathscr R=R[x_1,\ldots, x_s]/\langle t_1(x_1),\ldots, t_s(x_s) \rangle$, where $R$ is a chain ring and each $t_i(x_i)$ in $R[x_i]$ is a monic square-free polynomial for $i=1,2,\ldots, s$.

We briefly review the vast literature on polycyclic codes over some special classes of rings. 
Alkenani et al. \cite{Alkenani_2020} studied constacyclic codes over the ring $\mathbb F_q[u_1, u_2]/\langle u_{1}^2 - u_1, u_{2}^2 - u_2, u_1u_2 - u_2u_1 \rangle$. Dinh et al. \cite{Dinh_2020}  investigated constacyclic codes over the ring $\mathbb F_p + u_1 \mathbb F_p + \cdots + u_s \mathbb F_p$, subject to the conditions $u_i^2 = u_i$,  $u_i u_j = u_j u_i = 0$, and $p$ is an odd prime.
Dinh et al. \cite{Dinh_2017} explored the algebraic structure of constacyclic codes over finite semi-simple rings,   defined as direct sums of finite fields. Islam et al. \cite{Islam_2021}   focused their study on cyclic codes over the non-chain ring $\mathbb F_q[u]/\langle u^e-1\rangle$. Bag et.al \cite{Bag_2020}  examined skew constacyclic codes over the ring $\mathbb F_q [u, v]/\langle u^2- 1, v^2- 1, uv-vu\rangle$. Ali et al. \cite{Ali_2024}  worked out the case cyclic codes over non-chain rings $\mathbb  F_q[u_1, u_2, \ldots, u_s]/\langle u^2_ i - (\alpha_i)^2, \quad u_i u_j - u_j u_i\rangle$, where each $\alpha_i$ is a non-zero element of $\mathbb F_q$. Goyal and Raka \cite{Goyal_2020, Goyal_2021}   explored polyadic cyclic and constacyclic codes over the non-chain ring $\mathbb F_q[u,v]/ \langle f(u),g(v),uv -vu\rangle$, with $f(u)$ and $g(v)$ splitting into distinct linear factors over $\mathbb F_q$. Bhardwaj et al. \cite{Bhardwaj_2023} investigated constacyclic codes over the finite non-chain ring $\mathbb{F}_q[u_1,u_2, \ldots, u_k]/\langle f_i(u_i), u_i u_j - u_j u_i\rangle$, where each $f_i(u_i)$ is a polynomials that splits into distinct linear factors over $\mathbb{F}_q$. Qi \cite{Qi_2022} focused on polycyclic codes over $ \mathbb{F}_q+u\mathbb{F}_q $, with $u^2=0$. Karthick \cite{Karthick_2023} investigated  polycyclic codes over the ring $R=\mathbb F_q + u \mathbb F_q + v \mathbb F_q$, subject to the conditions $ u^2 = \alpha u$, $ v^2 = v$, and $ uv = vu = 0$, where $\alpha$ represents a unit element of $R$. 
 
The serial ring $\mathscr R=\frac{R[x_1,\ldots, x_s]}{\langle t_1(x_1),\ldots, t_s(x_s) \rangle}$, where $R$ is a finite commutative chain ring with maximal ideal $\mathfrak m$ and the natural projection of $t_i$ on $(R/\mathfrak m) [x_i]$ is a square free polynomial for each $i\in\{ 1,\ldots, s\}$,  is a comprehensive generalization of the rings previously discussed in the papers \cite{Ali_2024, Alkenani_2020, Bag_2020, Bhardwaj_2023,  Dinh_2020,   Goyal_2020, Goyal_2021, Islam_2021,  Karthick_2023,  Qi_2022}.  Note that $R$ is a Frobenius ring and when $R$ is assumed to be a field the serial ring $\mathscr{R}$ is a semi-simple ring, see Proposition \ref{32}. 

In this paper, we investigate polycyclic codes over the general ring $\mathscr{R}$. The outline of the paper is as follows.  Section \ref{77-1} shows the foundational structure of $\mathscr R$. In Section \ref{28-1}, we examine how a polycyclic code $\mathscr C$ over $\mathscr R$ decomposes into a direct sum $\oplus _{i=1}^k \mathscr C_i$, where each $\mathscr C_i$ constitutes a polycyclic code over a chain ring. 
Section \ref{7} explores the concept of quasi-2-dimensional polycyclic (Q2DP) codes along with their multidimensional counterparts (QsDP), which extend the ideas of quasi-polycyclic codes and multidimensional cyclic codes  \cite{Guneri_2016, Guneri_2012, Hassan_2022, Moro_2018, Wu_2022}. Furthermore, we define the concept of $f(x)$-polycyclic-QsDP codes as a specific subclass of QsDP codes. Then, we prove that when considering polycyclic codes over $\mathscr{R}$ as $R$-submodules, they precisely coincide with $f(x)$-polycyclic-QsDP codes.
 Polycyclic codes are defined as the invariant $\mathscr{R}$-submodules of $\mathscr{R}^n$ under the operator $\tau_{f(x)}$. Their counterpart are sequential codes, which are defined as invariant subspace under the adjoint operator of $\tau_{f(x)}$. It is well-known that if a linear subspace is invariant under one operator, its Euclidean dual is invariant under the adjoint operator. Consequently, a polycyclic code $\mathscr{C}$ and its Euclidean dual $\mathscr{C}^{\perp}$ do not share the same structure.  Thus, the first ones are equivalent to ideals in $\mathscr R[x]/\langle f(x)\rangle$, defining a polynomial and ideal structure for the second ones is challenging. Note that, if $\mathscr{C}$ is an ideal in $\mathscr R[x]/\langle f(x)\rangle$ generated by the polynomial $g(x)$ with the property that $f(x)=g(x)h(x)$, it cannot be assumed that $\mathscr{C}^{\perp}$ is generated by $h(x)$ because $\mathscr{C}^{\perp}$ is not necessarily a polycyclic code. This incorrect interpretation happened in the proof of Theorem 5 in \cite{Patel_2022}. In  Section \ref{74}, using multivariable codes over chain rings in \cite{Moro_2006, moro_2007}, we establish a polynomial structure for $\mathscr C^\perp$ for the special case where $f(x)$ is a polynomial with coefficients in $R$. The annihilator dual is an alternative duality concept ensuring that the dual of polycyclic codes has a polycyclic structure, a feature that fails in the Euclidean dual. In Section~\ref{66}, some properties of the annihilator dual of a polycyclic code over $\mathscr R$ are studied. Usually, it is stated in the literature that \textit{``a polycyclic code $\langle g(x) \rangle$ is annihilator self-dual if and only if $f(x) = g^2(x)$''}. As explained in Remark \ref{56}, that statement relies on the assumption that \textit{``if $b$ divides $c$ and $c$ divides $b$, then $b$ equals $c$'}', without considering that the divisibility relation for rings is different, namely \textit{``if $b$ divides $c$ and $c$ divides $b$, then $\langle b\rangle$ equals $\langle c\rangle$''}. Thus in Section~\ref{66}, we revise this property as follows, \textit{``a polycyclic code $\langle g(x) \rangle$ is annihilator self-dual if and only if $f(x)=ag^2(x)$ for some unit element $a$ in the ring $\mathscr R$''}. 
To finish this paper, we study some quantum code constructions from this type of codes. Calderbank-Shor-Steane (CSS) construction allows for constructing quantum codes from dual-preserving classical codes, see for example \cite{Nadella_2012, Sarvepalli_2008}. Existing literature has applied this construction only to codes that preserve Euclidean or Hermitian duality. In Section \ref{76}, we present this construction for polycyclic codes over Frobenius rings that preserve annihilator duality.
\section*{List of symbols and abbreviations}
\begin{itemize}

    \item[$R$]: Finite chain ring with the maximal ideal $\mathfrak{m}=\langle \gamma\rangle$ and the nilpotency index $e$.
    \item[$\mathbb{F}_q$]: Residue field $\mathbb F_q = R/\mathfrak{m}$, where $q = p^t$, $p$ a prime number.
    \item[$\overline{\phantom{r}}$]: Natural ring homomorphism from $R[y]$ to $\mathbb{F}_q[y]$.
    \item[$\mathscr{R}$]: Ring $R[x_1,\ldots, x_s]/\langle t_1(x_1),\ldots, t_s(x_s) \rangle$, where $t_i(x_i)\in R[x_i]$ is a monic square-free of degree $m_i$.
    \item[$H_i$]:  Set of roots of $\bar{t}_i(x_i)$ in an extension of $\mathbb{F}_q$.
    \item[$C(\nu)$]: Set of the class of $\nu \in \mathcal H = \prod_{i=1}^{s} H_i$ as $ \{(\nu^{ q^{j}_1}, \ldots, \nu^{ q^{j}_s}) \mid j \in \mathbb{N}\}$.
    \item[$\mathcal C$]: Set of all classes $C(\nu)$. 
    \item[$C$]: Simpler notation for $C(\nu)$.
    \item[$p_{C,i}(x_i)$]:  Polynomial $\text{Irr}(\nu_i, \mathbb{F}_q)$, where $\nu=(\nu_1, \ldots, \nu_s) \in \mathcal H $.
    \item[$b_{C,i}(x_i)$]: Polynomial $\text{Irr}(\nu_i, \mathbb{F}_q(\nu_1, \ldots, \nu_{i-1}))$.
    \item[$\widetilde{b}_{C,i}(x_i)$]: Polynomial  $p_{C,i}(x_i)/b_{C,i}(x_i)$.
    \item[$w_{C,i}(x_1, \ldots, x_i)$]: Polynomial obtained by replacing $\nu_i$ with $x_i$ in  $b_{C,i}(x_i)$.
    \item[$\pi_{C,i}(x_1, \ldots, x_i)$]: Polynomial obtained by replacing $\nu_i$ by $x_i$ in  $\widetilde{b}_{C,i}(x_i)$.
    \item[$q_{C,i}$]: Hensel lift to $R$ of $p_{C,i}$.
    \item[$z_{C,i}$]: Hensel lift to $R$ of $w_{C,i}$.
    \item[$\sigma_{C,i}$]: Hensel lift to $R$ of $\pi_{C,i}$.
    \item[$I_C$]: Ideal $ \langle q_{C,1}, z_{C,2}, \ldots, z_{C,s} \rangle$.
    \item[$T_C$]: Chain ring $R[x_1, \ldots, x_s]/I_C$.
    \item[$h_C$]: Polynomial  $h_C(x_1, x_2, \ldots, x_s)= \prod_{i=1}^{s} t_i(x_i) q_{C,i}(x_i) \prod_{i=1}^{s} \sigma_{q_{C,i}}(x_2, \ldots, x_i)$.
    \item[$e_C$]: Orthogonal idempotent such that $\mathscr R e_C  \cong \langle h_C + I \rangle$.
    \item[$\mathscr R_{i}$ ]: Chain ring isomorphic to  $\mathscr Re_{C_i}$, where $\{C_i\}_{i=1}^k$ is a set of representatives of $\mathcal C$.
    \item[$\mathscr{C}$]: Linear code of length $n$ over $\mathscr R$.
    \item[$f(x)$]: Polynomial $x^n-(f_{n-1}x^{n-1}+\cdots+f_1x+f_0)$ in $\mathscr R[x]$.
    \item[$\tau_{f(x)}$]: $f(x)$-polycyclic shift. 
    \item[$\boldsymbol c$]: Element $(c_0, c_1, \ldots, c_{n-1}) \in \mathscr R^n$.
    \item[$c_{i,j}$]: Components of $c_j$ ($0\leqslant j\leqslant n-1$) in \(c_j = c_{1,j} e_{C_1} + c_{2,j} e_{C_2}+ \cdots+ c_{k,j} e_{C_k}\in \mathscr R\).
    \item[$\boldsymbol c_i$]: Components of $\boldsymbol c$ in $\boldsymbol c=\boldsymbol c_1 e_{C_1}+ \boldsymbol c_2 e_{C_2}+\cdots+ \boldsymbol c_k e_{C_k}$.
    \item[$\mathscr C_{i}$]: Set  $\{ \boldsymbol c_i \in \mathscr R^n \mid \boldsymbol c_1 e_{C_1}+ \boldsymbol c_2 e_{C_2}+\cdots+ \boldsymbol c_k e_{C_k} \in \mathscr C\}$.
    \item[$f_i(x)$]: Polynomial $x^n - (f_{i,n-1}x^{n-1} + \cdots + f_{i,1}x + f_{i,0})$, where $f_{i,j}$ is component of $f_j$.
    \item[$\operatorname{rank}_{\mathscr R}(\mathscr C)$]: Minimum number of generators for spanning $\mathscr C$ as $\mathscr R$-linear combinations.
    \item[$\mathscr R_f$]: Quotient ring $\mathscr R[x]/\langle f(x)\rangle$.
    \item[$\phi$]: Map that associates each vector in $\mathscr{R}$ with its corresponding polynomial in $\mathscr{R}_f$.
    \item[$\mathbf x^\alpha$]: Monomial $x_1^{\alpha_1}  x_2^{\alpha_2} \ldots x_s^{\alpha_s}$, where $\alpha = (\alpha_1, \alpha_2, \ldots, \alpha_s)\in (\mathbb{N}\cup\{0\})^s$.
    \item[$\succ$]: Lexicographic ordering on the monomials.
    \item[$\Phi$]: $R$-isomorphism from  $\mathscr R$ to $R^{m}$, where $m=m_1m_2\ldots m_s$.
    \item[$QsDP$]: Quasi-s-dimensional polycyclic codes.
    \item[$C_f$]: Companion matrix of the polynomial $f$.
    \item[$I_{i\times i}$]: Identity matrix.
    \item[$\otimes$]: Tensor product of matrices.
    \item[$\nu_{f(x)}$]: $f(x)$-sequential shift.
    \item[$\mathscr{C}^\perp$]: Euclidean dual of $\mathscr{C}$.
    \item[$ \mathscr C^{\circ}$]: Annihilator dual of $\mathscr{C}$.
    \item[$LCD$]: Linear complementary dual code.
    \item[$\chi$]: Character in $\operatorname{Hom}(R; \mathbb{C}^*)$.
    \item[$\operatorname{Ann}(\mathscr C)$]: Annihilator of the subset $\mathscr C$ of $\mathscr R_f$.
    \item[$\operatorname{swt}$]: Symplectic weight.
    \item[$\operatorname{wt}$]: Hamming weight.

\end{itemize}
\section{The structure of the ring $\mathscr R=\frac{R[x_1,\ldots, x_s]}{\langle t_1(x_1),\ldots, t_s(x_s) \rangle}$}\label{77-1}
Let $R$ be a finite commutative chain ring, and let $\gamma \in R$ be a fixed generator of the maximal ideal $\mathfrak{m}$ with nilpotency index $e$. Consider $\mathbb F_q = R/\mathfrak{m}$ as the residue field, where $q = p^t$, $p$ a prime number. Let $\overline{\phantom{x}}: R[y] \to \mathbb{F}_q[y]$ represent the natural ring homomorphism, mapping $r \in R$ to $r + \mathfrak{m}$ and the variable $y$ to itself. A polynomial $p(y)=\sum_{i=0}^{n-1}p_iy^i\in R[y]$ is called square-free (basic irreducible) if $\bar{p}(y)$ is square-free (basic irreducible) in $\mathbb{F}_q[y]$. Additionally, it is called regular if it is not a zero divisor in $R[y]$. In \cite[Proposition 2.3]{Dinh_2004} it is stated that $p(y)$ is regular if and only if $p_i$ is a unit element for some $i$. Hence, monic polynomials are regular. If $g(y)\in R[y]$ is a regular basic irreducible polynomial, then the quotient ring $R[y]/\langle g(y)\rangle $ forms a chain ring, see \cite[Lemma 3.1]{Dinh_2004}. The monic square-free polynomial $g(y)$  factors uniquely as a product of monic basic irreducible pairwise coprime polynomials, given by  $g(y)=g_1(y)g_2(y)\cdots g_r(y)$ \cite[Theorem XIII.4]{McDonald_1974}. Hence, by the Chinese Remainder Theorem, the quotient ring $R[y]/\langle g(y) \rangle$ can be decomposed as a direct sum of chain rings
 \begin{equation}\label{1-2}
    \frac{R[y]}{\langle g(y) \rangle}\cong \frac{R[y]}{\langle g_1(y) \rangle}\bigoplus \frac{R[y]}{\langle g_2(y) \rangle} \bigoplus\cdots\bigoplus \frac{R[y]}{\langle g_r(y) \rangle}\cdot
 \end{equation} 
The algebraic properties of the algebra $\mathscr{R}=R[x_1,\ldots, x_s]/\langle t_1(x_1),\ldots, t_s(x_s) \rangle$, assuming that each $t_i(x_i)$ in $R[x_i]$ is a monic square-free polynomial of degree $m_i$, was explored in \cite{Moro_2006}. Let $H_i$ denote the set of roots of $\bar{t}_i(x_i)$ in a suitable extension of $\mathbb{F}_q$. For each $\nu \in \mathcal H = \prod_{i=1}^{s} H_i$, let $C(\nu)$ denote the set of the class of $\nu$ as $ \{(\nu^{ q^{j}_1}, \ldots, \nu^{ q^{j}_s}) \mid j \in \mathbb{N}\}$. The set of all these classes, denoted by $\mathcal C = \mathcal C(t_1, \ldots, t_s)$,  forms a partition of $ \mathcal H$. 

The size of each class is given by $|C(\nu)| = \text{lcm}(d_1, \ldots, d_s) = [\mathbb{F}_q(\nu_1, \ldots, \nu_s) : \mathbb{F}_q]$, where $d_i$ is the degree of the irreducible polynomial of $\nu_i$ over $\mathbb{F}_q$. Let $p_{C,i}(x_i)$ denote the polynomial $\text{Irr}(\nu_i, \mathbb{F}_q)$, where $\nu=(\nu_1, \ldots, \nu_s) \in \mathcal H $ and $C=C(\nu)$ is a class of $\nu$. For each $i = 2, \ldots, s$, let $b_{C,i}(x_i) = \text{Irr}(\nu_i, \mathbb{F}_q(\nu_1, \ldots, \nu_{i-1})) \in \mathbb{F}_q(\nu_1, \ldots, \nu_{i-1})[x_i]$ and $\widetilde{b}_{C,i}(x_i) = p_{C,i}(x_i)/b_{C,i}(x_i)$. The polynomials $b_{C,i}(x_i)$ and $\widetilde{b}_{C,i}(x_i)$ are coprime polynomials. One can define the multivariable polynomials $w_{C,i}(x_1, \ldots, x_i)$ and $\pi_{C,i}(x_1, \ldots, x_i)$ obtained by substituting $\nu_i$ by $x_i$ in polynomials $b_{C,i}(x_i)$ and $\widetilde{b}_{C,i}(x_i)$, respectively and denote the Hensel lifts to $R$ of the polynomials $p_{C,i}$, $w_{C,i}$, and $\pi_{C,i}$ by $q_{C,i}$, $z_{C,i}$, and $\sigma_{C,i}$, respectively. If we denote $ \langle q_{C,1}, z_{C,2}, \ldots, z_{C,s} \rangle$ by $I_C$, then the ring $T_C = R[x_1, \ldots, x_s]/I_C$ is a chain ring with maximal ideal $\mathfrak M = \langle a, q_{C,1}, z_{C,2}, \ldots, z_{C,s} \rangle + I_C$, and the residue field $T_C/\mathfrak  M \cong \mathbb{F}_q(\nu_1, \ldots, \nu_s)$. Let us define  the polynomial
 \begin{equation}\label{72}
    h_C(x_1, x_2, \ldots, x_s) := \prod_{i=1}^{s} t_i(x_i) q_{C,i}(x_i) \prod_{i=1}^{s} \sigma_{q_{C,i}}(x_2, \ldots, x_i).
 \end{equation}
The following two propositions outline the properties of the serial ring $\mathscr{R}$.
 
 \begin{Proposition}[\cite{Moro_2006}]\label{13}\
    \begin{enumerate}
        \item $\mathscr R \cong \bigoplus_{C \in \mathcal C} \langle h_C + I \rangle\cong  \bigoplus_{C \in \mathcal C} R[x_1, \ldots, x_s]/I_C$,  where 
        $\mathcal C$ is the partition of $\mathcal H$.
        \item $R[x_1, \ldots, x_s]/I_C$ is a chain ring with maximal ideal $\mathfrak M = \langle a, q_{C,1}, z_{C,2}, \ldots, z_{C,s} \rangle + I_C$ and its
              residue field is given by $T_C/\mathfrak  M \cong \mathbb{F}_q(\nu_1, \ldots, \nu_s)$.
        \item There exists a complete set of orthogonal idempotent elements $\{e_C \in \mathscr 
              R \mid C \in \mathcal C\}$ such that $\mathscr R e_C  \cong \langle h_C + I \rangle$.
        \item If $\lvert \mathcal{C} \rvert = k$ and $\{C_i\}_{i=1}^k$ is a set of representatives of $\mathcal C$, then the ring $\mathscr{R}$ can be decomposed as a direct sum $$\mathscr R\cong\bigoplus_{i=1}^k\mathscr R 
        e_{C_i}\cong \bigoplus_{i=1}^k\mathscr R_{i},$$ where each $\mathscr R_{i}$ is a chain ring. Moreover,
             any element $r\in \mathscr R$  can be expressed uniquely as  \(r = r_1e_{C_1} + r_2e_{C_2} +\cdots+ r_ke_{C_k}\), where $r_i\in \mathscr R$.            
    \end{enumerate} 
 \end{Proposition}

A ring $A$ is called  \emph{uniserial} if all $A$-submodules of the module $A_A$ are linearly ordered by inclusion.  A ring is called {serial} if it is a direct sum of uniserial rings. Any uniserial ring is local, and a serial ring is uniserial if and only if it is local, see \cite[ Section 1.4]{Puninski_2001}. Note that any chain ring is a serial ring, thus, if $g(y)\in R[y]$ is a monic basic irreducible polynomial, then the quotient ring $R[y]/\langle g(y)\rangle$ is a serial ring. Finite commutative semi-simple rings are those rings that can be represented as finite direct sums of finite fields. The following proposition is straightforward from the previous paragraph.

 \begin{Proposition}\label{32}\
  \begin{enumerate}
    \item $\mathscr R$  is a serial ring.
    \item $\mathscr R$  is a principal ideal Frobenius ring.
    \item If R is a field, then $\mathscr R$  is a semi-simple ring.  
  \end{enumerate} 
 \end{Proposition}

\section{Polycyclic codes over $\mathscr R$}\label{28-1}

A subset $\mathscr{C}\subseteq \mathscr R^n$  is called a \emph{linear code over} $\mathscr R$ (or an $\mathscr R$-linear code) of length $n$ if it is an $\mathscr R$-submodule of $\mathscr R^n$. It is called  \emph{free code} if it is a free $\mathscr{R}$-submodule. Throughout this paper 
\begin{equation}\label{79}
    f(x)=x^n-(f_{n-1}x^{n-1}+\cdots+f_1x+f_0)
\end{equation}
will denote a monic polynomial in $\mathscr R[x]$. An $f(x)$-polycyclic shift $\tau_{f(x)} : \mathscr{R}^n \rightarrow \mathscr{R}^n$ is defined by 
\begin{equation}\label{90}
\tau_{f(x)}(b_0, b_1, \ldots, b_{n-1}) = (0, b_0, b_1, \ldots, b_{n-2}) + b_{n-1}(f_0, f_1, \ldots, f_{n-1}).  
\end{equation}
An $\mathscr R$-linear code of length $n$ is called \emph{$f(x)$-polycyclic code} if it remains invariant under the $f(x)$-polycyclic shift $\tau_{f(x)}$.

According to Proposition \ref{13}, every component \(c_j\in \mathscr R\) in   \(\boldsymbol c = (c_0, c_1, \ldots, c_{n-1}) \in \mathscr C\subseteq \mathscr R^n\) has a unique representation  \(c_j = c_{1,j} e_{C_1} + c_{2,j} e_{C_2}+ \cdots+ c_{k,j} e_{C_k}\), where $c_{i,j}\in \mathscr R$ for all $i,j$, $1\leqslant i\leqslant k$,\, $0\leqslant j\leqslant n-1$. Therefore, $\boldsymbol c$ can be  expressed as
 \begin{equation}\label{33}
  \boldsymbol c = (c_{1,0}, c_{1,1}, c_{1,2},\ldots, c_{1,n-1})e_{C_1}+ (c_{2,0}, c_{2,1}, c_{2,2},\ldots, c_{2,n-1})e_{C_2}+\cdots+ (c_{k,0}, c_{k,1}, c_{k,2},\ldots, c_{k,n-1})e_{C_k},            
\end{equation}
which we  will denote  as $\boldsymbol c=\boldsymbol c_1 e_{C_1}+ \boldsymbol c_2 e_{C_2}+\cdots+ \boldsymbol c_k e_{C_k}$.
For a  code $\mathscr C$ of length $n$ over $\mathscr R$, we define the subsets $\mathscr{C}_{i}$ of $\mathscr{R}^{n}$ for $1 \leqslant i \leqslant k$ as follows:
\begin{equation*}
    \mathscr C_{i} = \{ \boldsymbol c_i \in \mathscr R^n \mid \boldsymbol c_1 e_{C_1}+ \boldsymbol c_2 e_{C_2}+\cdots+ \boldsymbol c_k e_{C_k} \in \mathscr C\}.
\end{equation*} The following result is straightforward.
 \begin{Theorem}
     Let $\mathscr C$ be a linear code of length $n$ over $\mathscr R$. Then, each $\mathscr C_i$ is a linear code over $\mathscr R$ and $\mathscr C \cong \oplus_{i=1}^k\mathscr C_ie_{C_i} $. 
 \end{Theorem}
Similarly to Equation~\eqref{33}, one can rewrite the polycyclic relationship as follows: consider the decomposition by the orthogonal idempotents  $\boldsymbol f_i = (f_{i,0}, f_{i,1}, \ldots, f_{i,n-1})\in\mathscr R^n$ with \(\boldsymbol f = (f_0, f_1, \ldots, f_{n-1})\in\mathscr R^n\), and one can write $\boldsymbol f=\boldsymbol f_1 e_{C_1}+ \boldsymbol f_2 e_{C_2}+\cdots+ \boldsymbol f_k e_{C_k}.$ Corresponding to each $\boldsymbol{f}_i$, we define the polynomial
\begin{equation}\label{34}
    f_i(x)= x^n - (f_{i,n-1}x^{n-1} + \cdots + f_{i,1}x + f_{i,0}).
\end{equation}
 \begin{Lemma}\label{3}
   Let $\mathscr C$ be a linear code over $\mathscr R$. Then, $\tau_{f(x)}( \boldsymbol c) \in \mathscr C$ if and only if $\tau_{f_i(x)}(\boldsymbol c_i)\in \mathscr C_i$ for all $1\leqslant i\leqslant k$.
 \end{Lemma}
 \begin{proof}
   Let $\boldsymbol c = (c_0, c_1, \ldots, c_{n-1}) \in \mathscr C$. Since $\{e_{C_1,},\ldots, e_{C_k}\}$ is a complete   set of orthogonal idempotents,  we have
   \begin{align*}
    \tau_{f(x)}(\boldsymbol c) &= (0, c_0, c_1, \ldots, c_{n-2}) + c_{n-1}( f_0, f_1, \ldots, f_{n-1})\\
                            &= \bigg(0, \sum_{i=1}^{k} c_{i,0}e_{C_i}, \sum_{i=1}^{k} c_{i,1}e_{C_i}, \ldots, \sum_{i=1}^{k} c_{i,n-2}e_{C_i}\bigg) + \sum_{i=1}^{k} c_{i,n-1}e_{C_i}\bigg(\sum_{i=1}^{k}f_{i,0}e_{C_i}, \sum_{i=1}^{k}f_{i,1}e_{C_i}, \ldots, \sum_{i=1}^{k}f_{i,n-1}e_{C_i}\bigg)\\
                            &= \bigg(0, \sum_{i=1}^{k} c_{i,0}e_{C_i}, \sum_{i=1}^{k} c_{i,1}e_{C_i}, \ldots, \sum_{i=1}^{k} c_{i,n-2}e_{C_i}\bigg) + \bigg(\sum_{i=1}^{k}c_{i,n-1}f_{i,0}e_{C_i}, \sum_{i=1}^{k}c_{i,n-1}f_{i,1}e_{C_i}, \ldots, \sum_{i=1}^{k}c_{i,n-1}f_{i,n-1}e_{C_i}\bigg)\\
                            &=  \sum_{i=1}^{k}\bigg( \big( 0, c_{i,0}, c_{i,1},\ldots, c_{i,n-2} \big)+c_{i,n-1}\big( f_{i,0}, f_{i,1}, \ldots, f_{i,n-1} \big)\bigg)e_{C_i} \\
                            &= \sum_{i=1}^{k}\bigg( \tau_{f_i(x)}(c_{i,0}, c_{i,1},\ldots, c_{i,n-1})\bigg)e_{C_i}=\sum_{i=1}^{k}\bigg(\tau_{f_i(x)}(\boldsymbol c_i)\bigg)e_{C_i}.
  \end{align*}
  Therefore, $\tau_{f(x)}( \boldsymbol c) \in \mathscr C$ if and only if $\tau_{f_i(x)}(\boldsymbol c_i)\in \mathscr C_i$ for all $1\leqslant i\leqslant k$.
 \end{proof}
One can consider the $\mathscr R$-module isomorphism $\phi:\mathscr R^n\to \mathscr R_f=\mathscr R[x]/\langle f(x)\rangle$, which maps each vector to its corresponding polynomial representation. It was proved in \cite{Dougherty_2017} that the image of any $f(x)$-polycyclic code under $\phi$ is an ideal in $\mathscr R_f$. From now on, we abuse the notation $\mathscr{C}$ and we will also call the ideals in the ring $\mathscr{R}_f$ $f(x)$-polycyclic codes. The ring  $\mathscr R_f$ might not always be a principal ideal ring and d a single codeword generator is a nice property for encoding and decoding. Therefore, we determine the condition under which the ring $\mathscr R_f$ is a principal ideal ring.

\begin{Proposition}\label{21}\
\begin{enumerate}
    \item The ring $\mathscr R[x]/\langle f(x)\rangle$ is isomorphic to $\oplus_{i=1}^k \mathscr R_i[x]/\langle f_i(x)\rangle$, where $\mathscr R_i$ represents the chain ring defined in Proposition \ref{13} and $f_i(x)$ denotes the polynomial defined in Equation~\eqref{34}.
    \item The ring $\mathscr R[x]/\langle f(x)\rangle$ is principal ideal if and only if  $\mathscr R_i[x]/\langle f_i(x)\rangle$ is a principal ideal ring for each $i\in \{1,\ldots ,k\}$.
\end{enumerate}  
 \end{Proposition}

 \begin{proof}\
\begin{enumerate}
    \item Let $g(x)=g_0+g_1x+\cdots+g_{n-1}x^{n-1}$ be an element in $\mathscr R[x]/\langle f(x)\rangle$. We know that the coefficients $g_j\in \mathscr R$ are uniquely expressed as $g_j = g_{1,j} e_{C_1} + g_{2,j} e_{C_2}+ \cdots+ g_{k,j} e_{C_k}$, i.e., $g_j$ can be uniquely expressed as $(g_{1,j},\, g_{2,j},\, \ldots,\, g_{k,j})$. Therefore, $g(x)=g_1(x)e_{C_1}+ g_2(x)e_{C_2}+\cdots+ g_k(x)e_{C_k}$, where $g_i(x)= g_{i,0}+ g_{i,1}x+ g_{i,2}x^2+\cdots+ g_{i,n-1}x^{n-1}\in\mathscr R_i[x]$. Hence, the polynomial $g(x)$ can also be represented as  $k$-tuples of polynomials $g(x)=\big(g_1(x), g_2(x),\ldots, g_k(x)\big)$.  Now, we define the ring isomorphism
    \begin{align}\label{24}
      \eta :\mathscr R_f=\frac{\mathscr R[x]}{\langle f(x)\rangle} &\longrightarrow \frac{\mathscr R_1[x]}{\langle f_1(x)\rangle}\oplus \frac{\mathscr R_2[x]}{\langle f_2(x)\rangle} \oplus \dots \oplus \frac{\mathscr R_k[x]}{\langle f_k(x)\rangle}, \\
      g(x) &\longmapsto \big(g_1(x),\, g_2(x),\, \dots,\, g_k(x)\big).\nonumber
    \end{align}
    \item It is straightforward.
\end{enumerate}
\end{proof}
\begin{Remark}\label{81}\
\begin{itemize}
\item The condition for the ring $\mathscr R_i[x]/\langle f_i(x)\rangle$ being a principal ideal ring in terms of the chain ring  $\mathscr R_i$ and the polynomial $f_i(x)$ can be found in \cite[Theorem 1]{Moro_2018b}.
\item If $f_i(x)$ is a regular basic irreducible polynomial in the ring $\mathscr R_i[x]$, then $\mathscr R_i[x]/\langle f_i(x)\rangle$ is a chain ring with the maximal ideal $\langle \gamma_i+\langle f_i(x)\rangle\rangle$, where $\gamma_i$ denotes the maximal ideal of the chain ring $\mathscr R_i$, see \cite[Lemma 3.1]{Dinh_2004}.
\end{itemize}
\end{Remark}
 Let $\mathscr C$ be an $f(x)$-polycyclic code over $\mathscr R$, we will denote the \emph{rank} of the code $\mathscr C$ as  $\operatorname{rank}_{\mathscr R}(\mathscr C)$, i.e. the minimum number of generators for spanning $\mathscr C$ as $\mathscr R$-linear combinations. We will also denote the \emph{minimum Hamming distance} of $\mathscr C$  as  $d(\mathscr C)$.
\begin{Theorem}\label{22} Let $\mathscr C$ be an $f(x)$-polycyclic code over $\mathscr R$.  Then we have the following statements.
 \begin{enumerate}
    \item There are $f_i(x)$-polycyclic codes $\mathscr C_i$ in $\mathscr R_i[x]/\langle f_i(x)\rangle$ such that  $\mathscr C\cong\oplus_{i=1}^k\mathscr C_i$. 
    \item $\mathscr C$ is an $f(x)$-polycyclic code if and only if every $\mathscr C_i$ is an $f_i(x)$-polycyclic code.
     \item $\mathscr C$ is generated by a single codeword if and only if every $\mathscr C_i$  is generated by a single codeword.
     \item If each $\mathscr C_i$ is generated by the polynomial $g_i(x)$, then $\mathscr C$ is generated by $g(x) = \oplus_{i=1}^k g_i(x)e_{C_i}$.
    \item $|\mathscr C| = \prod_{i=1}^{k} |\mathscr C_i|$.
    \item $ \operatorname{rank}_{\mathscr R}(\mathscr C) = \max\{\operatorname{rank}_{\mathscr R_i}(\mathscr C_i)\mid i = 1, \ldots, k\}$.
    \item $d(\mathscr C)=\min\{d(\mathscr C_i)\mid i = 1, \ldots, k\}$.
    \end{enumerate} 
\end{Theorem}

\begin{proof}\
\begin{enumerate}
    \item The ideal $\mathscr C$ in $\mathscr R[x]/\langle f(x)\rangle$ can be expressed as $\mathscr C = \eta(\mathscr C_1,  \mathscr C_2, \ldots,  \mathscr C_k)$, where each $\mathscr C_i$ is an ideal in $\mathscr R_i[x]/\langle f_i(x)\rangle$ and  $\eta$ is the isomorphism defined in \eqref{24}. Conversely, the set of ideals $\mathscr C_i$ in  $\mathscr R_i[x]/\langle f_i(x)\rangle$ corresponds to an ideal $\mathscr C := \eta^{-1}(\mathscr C_1, \mathscr C_2, \ldots, \mathscr C_k)$ in $\mathscr C$ in $\mathscr R[x]/\langle f(x)\rangle$.
    \item It follows directly from the proof of the first statement.
     \item We apply the second statement in Proposition \ref{21}.
     \item It follows from the proof of Proposition \ref{21}.
    \item It is straightforward from the first statement.
    \item Let $r_i=\operatorname{rank}_{\mathscr R_i}(\mathscr C_i)$ and $w_{1}^i, w_{2}^i, \ldots, w_{r_i}^i$ be a set of generators for $\mathscr C_i$. Let $r=\max\{r_1,r_2,\ldots, r_k\}$. Consider the set of vectors
      \begin{equation*}
       \left\{\eta^{-1}(v_{1}^1, v_{1}^2, \ldots, v_{1}^k), \eta^{-1}(v_{2}^1, v_{2}^2, \ldots, v_{2}^k), \ldots, \eta^{-1}(v_{r}^1, v_{r}^2, \ldots, v_{r}^k)\right\}
       \end{equation*}
      where $v_j^i$ is the $j$-th component of  $w_j^i$ and zeros are added as necessary to make each vector have $k$ components. This set spans $\mathscr C$. To see this, note that each generator $w_j^i$ can be mapped via $\eta$ to a vector in $\mathscr C$ by placing it in the $i$-th component and filling the other components with zeros. Thus, this set generates $\mathscr C$. Since we require $r$ vectors to generate $\mathscr C$ and there exists at least one index $i$ such that $r_i=r$, therefore it follows that $\operatorname{rank}_{\mathscr R}(\mathscr C) = \max\{\operatorname{rank}_{\mathscr R_i}(\mathscr C_i)\mid i = 1, \ldots, k\}$.
    \item Let $d_\star=\min\{d(\mathscr C_i)\mid i = 1, \ldots, k\}$. Then, there exists a $j\in \{1,\ldots ,k\}$ such that $d(\mathscr C_j) = d_\star$. Let $v_j$ be a minimum weight non-zero vector in $\mathscr C_j$. Consider the vector $\eta^{-1}(0, 0, \cdots, 0, v_j, 0, \ldots, 0)$, where $v_j$ is in the $j$-th component. This vector has Hamming weight $d_\star$, which implies $d(\mathscr C) \leqslant d_\star$.
       Now, let $v$ be a minimum weight vector in $\mathscr C$. Then $\eta(v)$ has components $(v_1, v_2, \ldots, v_k)$, where each $v_i$ belongs to $\mathscr C_i$. Since $d_\star$ is the minimum distance of $\mathscr C_j$, the weight of $v_j$ is at least $d_\star$. Thus, the weight of $\eta(v)$ is at least $d_\star$, which implies $d(\mathscr C) \geqslant d_\star$. Therefore, $d_\star = d(\mathscr C)$.  
   \end{enumerate}
\end{proof}

\section{Description of polycyclic codes over $\mathscr R$ as $R$-submodules}\label{7}

 A monomial $\mathbf x^\alpha$ in $R[x_1,x_2,\ldots, x_s]$, where $\alpha = (\alpha_1, \alpha_2, \ldots, \alpha_s)\in (\mathbb{N}\cup\{0\})^s$,  is a product of the form $\mathbf x^\alpha= x_1^{\alpha_1}  x_2^{\alpha_2} \ldots x_s^{\alpha_s}$.   We fix the lexicographic ordering on the monomials induced from the lexicographic ordering on $(\mathbb{N}\cup\{0\})^s$, that is,  $\mathbf{x}^{\alpha}{\succ} \mathbf{x}^{\beta}$  if the leftmost nonzero entry in the vector difference $\alpha - \beta $ is positive. Any element $g(x_1, x_2, \ldots, x_s)\in R[x_1, x_2, \ldots, x_s]$ can be expressed uniquely as 
 \begin{equation}\label{11}
   g(x_1, x_2, \ldots, x_s) = \sum\limits_{\alpha \in (\mathbb{N}\cup\{0\})^s} g_{\alpha} \mathbf{x}^\alpha,
 \end{equation}
where finitely many coefficients $g_\alpha$ are non-zero.
The degree of the monomial $\mathbf{x}^\alpha$ is   $\alpha_1 + \alpha_2 + \cdots + \alpha_s$ and the degree of $ g(x_1, x_2, \ldots, x_s)$ is the maximum degree among the degrees of the monomials involved in $ g(x_1, x_2, \ldots, x_s)$. Note that, fixed the lexicographic ordering, the sorted list of monomials  can be represented as $\{\mathbf{x}^{\alpha^{(1)}}, \mathbf{
x}^{\alpha^{(2)}}, \ldots, \mathbf{x}^{\alpha^{(m)}}\}$, where $m=m_1m_2\ldots m_s$ and $\alpha^{(j)}_i < m_i$ the degree of $t_i(x_i)$ is a basis for  $\mathscr R$. Therefore, we can represent an element $ g(x_1, x_2, \ldots, x_k)\in \mathscr R$  as $\sum_{i=1}^{m} g_{i} \mathbf{x}^{\alpha^{(i)}}$.  We define the $R$-isomorphism $\Phi: \mathscr R \to R^{m}$, which maps   $g(x_1, x_2, \ldots, x_s)$ to the vector $(g_1, g_2,\ldots, g_{m})$. Therefore, $\Psi:=\Phi\circ \phi^{-1}: \mathscr R_f\to R^{mn}$ is an $R$-module isomorphism, where $\phi$ is defined in Section \ref{28-1}. Through this mapping, $f(x)$-polycyclic codes over $\mathscr R$ can be viewed as $R$-submodules of $R^{mn}$. Consider $\mathscr C$   an $f(x)$-polycyclic code, and let $c(x)=c_0+c_1x+c_2x^2+\cdots+c_{n-1}x^{n-1}$ be an element of $\mathscr C$. When $s=2$, we take the lexicographical order 
\begin{equation}
    \{1,\,x_1,\,\ldots\,, x_1^{m_1-1},\, x_2,\, x_1x_2,\, \ldots\,, x_1^{m_1-1}x_2,\, \ldots\,, x_2^{m_2-1}, \,x_1x_2^{m_2-1},\, \ldots\,, x_1^{m_1-1}x_2^{m_2-1}\}.
\end{equation}
Hence, $\Psi(c)$ is represented as the vector $(\Phi(c_0), \Phi(c_1), \ldots, \Phi(c_{n-1}))$, where
\begin{equation}
    \Phi(c_i)= (c^{0,0}_i, \ldots, c^{m_1-1,0}_i, c^{0,1}_i, \ldots, c^{m_1-1,1}_i, \ldots, c^{0,m_2-1}_i, \ldots, c^{m_1-1,m_2-1}_i)
\end{equation}
and $c^{u,v}_i$ are the coefficients of $x_1^ux_2^v$ for $0\leqslant u\leqslant m_1-1$  and $0\leqslant v\leqslant m_2-1$.
Each $\Phi(c_i)$  can be viewed as the following array

 \begin{equation}\label{8}
     \begin{pmatrix}
    c^{0,0}_i & \cdots & c^{m_1-1,0}_i \\
    c^{0,1}_i & \cdots & c^{m_1-1,1}_i \\
    \vdots & \ddots & \vdots \\
    c^{0,m_2-1}_i & \cdots & c^{m_1-1,m_2-1}_i \\
    \end{pmatrix}.
\end{equation}
Therefore, $\Psi(c)$ consist of $n$-subarrays arranged as follows:

\begin{equation}\label{61}
     \begin{pmatrix}
    \begin{array}{c|c|c|c}
        \begin{matrix}
            c^{0,0}_0 & \cdots & c^{m_1-1,0}_0 \\
            c^{0,1}_0 & \cdots & c^{m_1-1,1}_0 \\
            \vdots & \ddots & \vdots \\
            c^{0,m_2-1}_0 & \cdots & c^{m_1-1,m_2-1}_0 \\
        \end{matrix} & 
        \begin{matrix}
            c^{0,0}_1 & \cdots & c^{m_1-1,0}_1 \\
            c^{0,1}_1 & \cdots & c^{m_1-1,1}_1 \\
            \vdots & \ddots & \vdots \\
            c^{0,m_2-1}_1 & \cdots & c^{m_1-1,m_2-1}_1 \\
        \end{matrix} & \,\,\cdots \,\, &
        \begin{matrix}
            c^{0,0}_{n-1} & \cdots & c^{m_1-1,0}_{n-1} \\
            c^{0,1}_{n-1} & \cdots & c^{m_1-1,1}_{n-1} \\
            \vdots & \ddots & \vdots \\
            c^{0,m_2-1}_{n-1} & \cdots & c^{m_1-1,m_2-1}_{n-1} \\
        \end{matrix} \\
    \end{array}
\end{pmatrix}.
\end{equation}
\begin{Remark}\label{50}
    Since $\mathscr{C}$ is an $f(x)$-polycyclic code, it follows that $x_1c\in \mathscr{C}$, $x_2c\in \mathscr{C}$, and $xc\in \mathscr{C}$. The inclusion of $x_1c\in \mathscr{C}$ implies that the columns of each subarray in \eqref{61} are closed under $t_1(x_1)$-polycyclic shift. Similarly, the inclusion of $x_2c$ in $\mathscr{C}$ implies that the rows of each subarray in Equation \eqref{61} are closed under $t_2(x_2)$-polycyclic shift. Furthermore, the inclusion of $xc\in \mathscr{C}$ implies that the subarrays in Equation \eqref{61} themselves are closed under $f(x)$-polycyclic shift.
\end{Remark}

\begin{Definition}
   Let $m=m_1 m_2$. A length $m n$ linear code $C$ over the chain ring $R$ with codewords given as in Equation \eqref{61} is a called quasi-2-dimensional polycyclic (Q2DP) code if the columns of each subarray are closed under $t_1(x_1)$-polycyclic shift and the rows of each subarray are closed under $t_2(x_2)$-polycyclic shift.
\end{Definition}
The above definition is a generalization of quasi-2-dimensional cyclic (Q2DC) codes defined in \cite{Guneri_2016}. Similarly, we can define a multidimensional version of Q2DP codes and abbreviate them as QsDP.
\begin{Definition}
   A Q2DP code is called an $f(x)$-polycyclic-Q2DP code if the subarrays in Equation \eqref{61} are closed under the $f(x)$-polycyclic shift.
\end{Definition}
Similarly, we can define a multidimensional version of $f(x)$-polycyclic-Q2DP codes, termed as $f(x)$-polycyclic-QsDP codes. Based on Remark \ref{50}, we have the following theorem.
\begin{Theorem}\label{77}
  Any $f(x)$-polycyclic code over $\mathscr R$ is $R$-isomorphic to an $f(x)$-polycyclic-QsDP code.
\end{Theorem}
We know that the basis of the algebra $\mathscr R_f$ as an $\mathscr R$-module is $\{1, x,\ldots, x^{n-1}\}$. Recall that $\phi^{-1}$ maps the polynomial $g(x)=g_0+g_1x+\cdots+g_{n-1}x^{n-1}\in \mathscr R_f$ to $(g_0,g_1,\ldots, g_{n-1})$. The \emph{representation matrix} of $g(x)$ is a matrix whose rows are  $\phi^{-1}(g(x)), \phi^{-1}(xg(x)), \ldots, \phi^{-1}(x^{n-1}g(x))$.  It is easy to see that the representation matrix of $x$ is the companion matrix $C_f$. So, similar to \cite{Bajalan_2022, Permouth_2009}, we can prove that any $f(x)$-polycyclic code over $\mathscr{R}$, as an $\mathscr{R}$-submodule of $\mathscr{R}^n$,  is invariant under right multiplication by the companion matrix $C_f$. Now, our objective is to explore the relationship between $f(x)$-polycyclic codes (or equivalently, $f(x)$-polycyclic-QsDP codes) and invariant $R$-submodules of $ R^{mn}$.

The basis of $\mathscr R_f$ as an $\mathscr R$-module is $\{1, x,\ldots, x^{n-1}\}$ whereas its basis as an $R$-module is the set

\begin{equation}
   \mathcal B= \{  \mathbf{x}^{\alpha^{(j)}} , \mathbf{x}^{\alpha^{(j)}} x, \ldots, \mathbf{x}^{\alpha^{(j)}} x^{n-1} \,:\, 1\leqslant j\leqslant m\}.
\end{equation}
For example, suppose that $\mathscr R=R[x_1,x_2]/\langle t_1(x_1), t_2(x_2)\rangle$, where $t_1(x_1)=x_1^2-b_1x_1-b_0$,  $t_2(x_2)=x_2^2-c_1x_2-c_0$ and $f(x)=x^2-f_1x-f_0$. If we take the lexicographic order $\{1,x_1,x_2,x_1x_2\}$ over $\mathscr R$, then the basis of  $\mathscr R_f$ as an $R$-module is $\mathcal B=\{1, x_1, x_2, x_1x_2, x, x_1x, x_2x, x_1x_2x \}$. The representation matrix of $x$ with respect to $\mathcal B$ is a matrix whose rows consist of coefficients of polynomials $x, x_1x, x_2x, x_1x_2x, x^2, x_1x^2, x_2x^2, x_1x_2x^2$, respectively. 
We denote this matrix by  $C_F$ and let $\Phi(f_i)=(r_0^i, r_1^i, r_2^i, r_3^i)$ for $i=0,1$. Then, it is easy to see that 
\begin{equation*}
   C_F=\begin{pmatrix}
         O_{4\times 4} & I_{4\times 4}\\
         F_0 & F_1\\
       \end{pmatrix},
 \end{equation*}
 where $O_{4\times 4}$ denotes the zero $4\times 4$ matrix, $I_{4\times 4}$  the identity $4\times 4$ matrix and
 \begin{equation*}
    F_0=\begin{pmatrix}
        {\Phi(f_0)\bigg(I_{2\times 2}\otimes I_{2\times 2}\bigg)}\\
        {\Phi(f_0)\bigg(I_{2\times 2}\otimes C_{t_1}\bigg)} \\     
        {\Phi(f_0)\bigg(C_{t_2}\otimes I_{2\times 2}\bigg)}\\      
        {\Phi(f_0)\bigg(C_{t_2}\otimes C_{t_1}\bigg)}\\         
       \end{pmatrix},\quad
    F_1=\begin{pmatrix}
        {\Phi(f_1)\bigg(I_{2\times 2}\otimes I_{2\times 2}\bigg)}\\
        {\Phi(f_1)\bigg(I_{2\times 2}\otimes C_{t_1}\bigg)} \\     
        {\Phi(f_1)\bigg(C_{t_2}\otimes I_{2\times 2}\bigg)}\\      
        {\Phi(f_1)\bigg(C_{t_2}\otimes C_{t_1}\bigg)}\\         
       \end{pmatrix}.
 \end{equation*}
Here $C_{t_1}$ and $C_{t_2}$ are companion matrices associated  with $t_1(x_1)$ and $t_2(x_2)$, respectively. We notice that 
 \begin{equation*}
     O_{4\times 4}=\begin{pmatrix}
         {\Phi(0)\bigg(I_{2\times 2}\otimes I_{2\times 2}\bigg)}\\
         {\Phi(0)\bigg(I_{2\times 2}\otimes C_{t_1}\bigg)} \\     
         {\Phi(0)\bigg(C_{t_2}\otimes I_{2\times 2}\bigg)}\\      
         {\Phi(0)\bigg(C_{t_2}\otimes C_{t_1}\bigg)}\\         
       \end{pmatrix},\quad
    I_{4\times 4}=\begin{pmatrix}
        {\Phi(1)\bigg(I_{2\times 2}\otimes I_{2\times 2}\bigg)}\\
        {\Phi(1)\bigg(I_{2\times 2}\otimes C_{t_1}\bigg)} \\     
        {\Phi(1)\bigg(C_{t_2}\otimes I_{2\times 2}\bigg)}\\      
        {\Phi(1)\bigg(C_{t_2}\otimes C_{t_1}\bigg)}\\         
      \end{pmatrix}.
      \end{equation*}
If we denote by $\Phi(C_f)$ the matrix obtained by applying the map $\Phi$ to each entry of the companion matrix $C_f$, it becomes evident that rows 1 and 5 of $C_F$ form $\Phi(C_f)\otimes \big( I_{2\times 2}\otimes I_{2\times 2} \big)$, rows 2 and 6 form $\Phi(C_f)\otimes \big(I_{2\times 2}\otimes C_{t_1}\big)$, rows 3 and 7 form $\Phi(C_f)\otimes \big(C_{t_2}\otimes I_{2\times 2}\big)$, and rows 4 and 8 form $\Phi(C_f)\otimes\big(C_{t_2}\otimes C_{t_1}\big)$. Therefore, $C_F$ is equivalent to the following matrix:

\begin{equation}\label{30}
    \begin{pmatrix}
        {\Phi(C_f)\bigg(I_{2\times 2}\otimes I_{2\times 2}\bigg)}\\
        {\Phi(C_f)\bigg(I_{2\times 2}\otimes C_{t_1}\bigg)} \\     
        {\Phi(C_f)\bigg(C_{t_2}\otimes I_{2\times 2}\bigg)}\\      
        {\Phi(C_f)\bigg(C_{t_2}\otimes C_{t_1}\bigg)}\\         
       \end{pmatrix}= \Phi(C_f)\otimes
       \begin{pmatrix}
        {I_{2\times 2}\otimes I_{2\times 2}}\\
        {I_{2\times 2}\otimes C_{t_1}} \\     
        {C_{t_2}\otimes I_{2\times 2}}\\      
        {C_{t_2}\otimes C_{t_1}}\\ 
        \end{pmatrix}= \Phi(C_f)\otimes\begin{pmatrix}
        {C_{t_2}^0\otimes C_{t_1}^0}\\
        {C_{t_2}^0\otimes C_{t_1}} \\     
        {C_{t_2}\otimes C_{t_1}^0}\\      
        {C_{t_2}\otimes C_{t_1}}\\ 
        \end{pmatrix}:=\Phi(C_f)\otimes M.
\end{equation}
Thus, if we consider $\mathscr R_f$ as an $R$-module, then the representation matrix of $x$ is the matrix in Equation \eqref{30}. Hence,  keeping the above notations, the following results hold for the $s=2$ case. 

\begin{Theorem}
Let $\mathscr C$ be an $f(x)$-polycyclic code over $\mathscr R$. Then $\Psi(\mathscr C)$ is closed under the right multiplication by the matrix $\Phi(C_f)\otimes M$, defined in \eqref{30}.
\end{Theorem}

\begin{Corollary}
    Any $f(x)$-polycyclic-Q2DP code is closed under the right multiplication by $\Phi(C_F)\otimes M$.
\end{Corollary}

Similarly, the above theorem and corollary are true for any value of $s$. It is worth noting that $\{1,x_1,x_2,x_1x_2\}$ forms the basis for $\mathscr R$, and the representation matrices for these basis elements are $\{C_{t_2}^0\otimes C_{t_1}^0, C_{t_2}^0\otimes C_{t_1}, C_{t_2}\otimes C_{t_1}^0, C_{t_2}\otimes C_{t_1}\}$, see \cite[Section 6]{Bajalan_2022}. We notice that the matrix $M$ is constructed based on this set. Therefore, for the general case $s$, the matrix $M$ can be constructed using the representation matrices of the basis elements of $\mathscr R$ and $R$-modules.

\section{Euclidean dual of polycyclic codes over $\mathscr R$}\label{74}

A linear code $\mathscr C$ of length $n$ over $\mathscr R$ is called \emph{$f(x)$-sequential} if it is invariant under the $f(x)$-sequential shift $\nu_{f(x)} : \mathscr{R}^n \rightarrow \mathscr{R}^n$ defined by 
\begin{equation}
    \nu_{f(x)}(b_0, b_1, \ldots, b_{n-1}) = (b_1,b_2,\ldots, b_{n-2},d),
\end{equation}
where $d=f_0b_0+f_1b_1+\cdots+f_{n-1}b_{n-1}$. 
As usual,  the \emph{Euclidean product} of $x$ and $y$ in $\mathscr{R}^n$ is given by $\langle x, y \rangle = x \cdot y^{\text{tr}}$, where $y^{\text{tr}}$ denotes the transpose of $y$. The Euclidean dual code $\mathscr{C}^\perp$ of an $\mathscr{R}$-linear code $\mathscr{C}$ of length $n$ is defined as $\mathscr{C}^\perp := \{x \in \mathscr{R}^n \mid \langle x, y \rangle = 0,\, \forall y \in \mathscr{C}\}$.
It is easy to see that the operator  $\nu_{f(x)}$ is the adjoint operator of $\tau_{f(x)}$, that is $\langle \tau_{f(x)}(x), y \rangle= \langle x, \nu_{f(x)}(y) \rangle$. It is well-known that if a linear subspace is invariant under one operator, its Euclidean dual is invariant under the adjoint operator. Therefore, we have the following theorem.  
\begin{Theorem}\label{74}
The linear code $\mathscr C$ is an $f(x)$-polycyclic code over $\mathscr R$ if and only if $\mathscr C^{\perp}$ is an $f(x)$-sequential code.
\end{Theorem}

\begin{Theorem}{\cite[ Theorem 3.5]{Permouth_2009}}
   The codes $\mathscr C$ and $\mathscr C^{\perp}$ are $f(x)$-polycyclic codes if and only if they are both $f(x)$-sequential codes, if and only if $\mathscr C$ is a constacyclic code.
\end{Theorem}
Using the results in \cite{Moro_2006},    we can establish a polynomial structure for $\mathscr C^\perp$ for the special case of $f(x)$. Recall that $R$ is a chain ring with the maximal ideal $\gamma$ and nilpotency index $e$,  $\mathcal C$ is the partition of $\mathcal H=\prod_{i=1}^{s} H_i$, and $h_C$ is defined in Equation \eqref{71}. We denote the ideal $\langle t_1(x_1),\ldots, t_s(x_s) \rangle\subseteq R[x_1,\ldots, x_s]$ by $I$.

\begin{Theorem}\label{78}
    Let $f(x)$ be a regular basic irreducible polynomial with coefficients in $R$. An $f(x)$-polycyclic code over the ring $\mathscr R$ is isomorphic to a sum of ideals of $R_f[x_1,\ldots, x_s]/I $ of the form  
    \begin{equation*}
        \left\langle \big(\gamma+\langle f(x)\rangle\big) ^{j_C}h_C(x_1, x_2, \ldots, x_s) + I\right\rangle,
    \end{equation*}
   where $0 \leqslant j_C \leqslant e$ and $R_f=R[x]/\langle f(x)\rangle$.
\end{Theorem}

\begin{proof}
 If $f(x)$ is a regular basic irreducible polynomial in the ring $R[x]$, then $R_f:=R/\langle f(x)\rangle$ is a chain ring with the maximal ideal $\langle \gamma+\langle f(x)\rangle\rangle$, see Remark \ref{81}.  It is well-known that 
\begin{equation}
  \frac{\mathscr R[x]}{\langle f(x)\rangle}\cong \frac{\frac{R[x]}{\langle f(x)\rangle}[x_1,\ldots, x_s]}{\langle t_1(x_1),\ldots, t_s(x_s) \rangle}=\frac{R_f[x_1,\ldots, x_s]}{\langle t_1(x_1),\ldots, t_s(x_s) \rangle}.
\end{equation}
Hence, every $f(x)$-polycyclic code over $\mathscr R$ is isomorphic to an ideal in the multivariable polynomial ring $R_f[x_1,\ldots, x_s]/\langle t_1(x_1),\ldots, t_s(x_s) \rangle $ over the chain ring $R_f$.
Now use Corollary 1 in \cite{Moro_2006}.
\end{proof}

\begin{Remark}
If $f(x)$ is a monic square-free polynomial with coefficients in $R$, it can be factored as a product of regular basic irreducible pairwise coprime polynomials with coefficients in $R$. Therefore, by the Chinese Remainder Theorem, we can extend the above theorem for a monic square-free polynomial with coefficients in $R$.  
\end{Remark}
\begin{Theorem}
     Let $f(x)$ be a regular basic irreducible polynomial with coefficients in $R$. Then the following statements hold.
    \begin{enumerate}
        \item The number of  $f(x)$-polycyclic codes over $\mathscr R$ is equal to $(e+1)^k$, where  $k=\lvert \mathcal{C} \rvert $.
        \item For the $f(x)$-polycyclic code $\mathscr C$, there exists a family of polynomials $G_0, \ldots, G_s \in R[x_1, \ldots, x_s]$ such that
    \begin{equation}\label{73}
        \mathscr C \cong \left\langle G_1,\,\,   \big( \gamma+\langle f(x)\rangle\big) G_2,\,\,  \ldots\,\, ,  \big( \gamma+\langle f(x)\rangle\big)^{e-1} G_s\right\rangle + I.
    \end{equation}
    \item  The size of the $f(x)$-polycyclic code $\mathscr C$ is equal to  
    \begin{equation*}
        \bigg|  \frac{R_f}{\langle \gamma+\langle f(x)\rangle\rangle}\Bigg|^{\sum _{i=0}^{e-1}(e-i)N_i},
    \end{equation*}
    where $N_i$ denotes the number of zeros $\mu \in \prod_{i=1}^{s} H_i$ of $\overline{G}_i$. Here, $\overline{G}_i$  represents the natural ring homomorphism image over \(R_f/\langle \gamma + \langle f(x) \rangle\rangle\). 
\end{enumerate}
\end{Theorem}

\begin{proof}
  Based on the proof of Theorem \ref{78} , every $f(x)$-polycyclic code over $\mathscr R$ is isomorphic to an ideal in the multivariable polynomial ring $R_f[x_1,\ldots, x_s]/\langle t_1(x_1),\ldots, t_s(x_s) \rangle $ over the chain ring $R_f$. Now, we apply Corollary 2, Theorem 3 and Corollary 3 in \cite{Moro_2006}. 
\end{proof}

\begin{Theorem}[$\mathscr R$ Abelian]
   Let $f(x)$ be a regular basic irreducible polynomial with coefficients in $R$, and let  $t_i(x_i)$ be square-free polynomials of the form  $t_i(x_i)=x_i^{e_i}-1$ for  $1\leqslant i\leqslant s$, and the ring $\mathscr R= R[x_1,\ldots, x_s]/\langle t_1(x_1),\ldots, t_s(x_s)\rangle$.   Then the following statements hold.  
  \begin{enumerate}
      \item If $f(x)$-polycyclic code $\mathscr C$ has the form in Equation \eqref{73}, then 
      \begin{equation}
          \mathscr C^\perp =  \left\langle \tau(G_0), \,\, \big( \gamma+\langle f(x)\rangle\big)\tau(G_2), \,\, \ldots\,\,, \big( \gamma+\langle f(x)\rangle\big)^{e-1}\tau(G_s)\right\rangle  + I,
      \end{equation}
      where $\tau$ is the ring automorphism of $R_f[x_1,\ldots, x_s]/\langle t_1(x_1),\ldots, t_s(x_s) \rangle $ given by $\tau(g(x_1,x_2,\ldots, x_s))=g(x_1^{e_1-1}, x_1^{e_2-1}, \ldots, x_1^{e_s-1} )$.
      \item The size of  $\mathscr C^{\perp}$ is equal to
      \begin{equation*}
          \bigg|  \frac{R_f}{\langle \gamma+\langle f(x)\rangle\rangle}\bigg|^{\sum _{i=0}^{e-1}iN_i},
      \end{equation*}
      where $N_i$ denotes the number of zeros $\mu \in \prod_{i=1}^{s} H_i$ of $\overline{G}_i$.
       \end{enumerate}
\end{Theorem}

\begin{proof}
   See  Theorem 4, Corollary 5 in \cite{Moro_2006}.
\end{proof}
\section{Annihilator dual of polycyclic codes over $\mathscr R$}\label{66}
As pointed before, the Euclidean dual of an $f(x)$-polycyclic code over $\mathscr{R}$ is not necessarily an $f(x)$-polycyclic code. In \cite{Alahmadi_2016}, an alternative duality for $f(x)$-polycyclic codes over fields is introduced, where the dual code of any $f(x)$-polycyclic code preserves the $f(x)$-polycyclic structure. In this section, we consider this duality for $f(x)$-polycyclic codes over $\mathscr{R}$. Let $A$ be a commutative ring with $1$ and $k(x)$, $l(x)$ be two non-zero polynomials in $A[x]$ such that the leading coefficient of $l(x)$ is a unit in $A$. Then the division algorithm holds, that is, there exist unique polynomials $q(x)$, $r(x) \in A[x]$ such that $ k(x) = q(x)l(x) + r(x)$ and either $r(x) = 0$ or $\deg r(x) < \deg l(x)$.
\begin{Definition}
 Let $\mathscr C$ be an $f(x)$-polycyclic code of length $n$ over $\mathscr R$. The annihilator product of $g(x), h(x) \in \mathscr R_f$ is defined as
  \begin{equation*}
    \langle g(x), h(x) \rangle_f = r(0),
  \end{equation*} 
 where  $r(x)$ is the remainder of $g(x)h(x)$ divided by $f(x)$. Moreover, the annihilator dual of $\mathscr C$ is defined by
$ \mathscr C^{\circ} = \{h(x) \in \mathscr R_f \mid \langle g(x), h(x) \rangle_f = 0, \forall g(x) \in \mathscr C\}. $ 
\end{Definition}
We notice that the definition is well-defined due to the uniqueness of the reminder in the division algorithm. The $f(x)$-polycyclic code $\mathscr C$ is called  \emph{annihilator self-orthogonal} (respectively \emph{annihilator self-dual},  \emph{annihilator linear complementary dual} (LCD), or \emph{annihilator dual-containing}) if $\mathscr C \subseteq \mathscr C^{\circ}$ (respectively $\mathscr C = \mathscr C^{\circ}$, $\mathscr C \cap \mathscr C^{\circ} = \{0\}$, or $\mathscr C^{\circ}\subseteq \mathscr C$). Recall that $$\operatorname{Ann}(\mathscr C)=\{h(x)\in \mathscr R_f \mid  g(x)h(x)=0 \,(\operatorname{mod}f(x)),\,\, \forall g(x)\in \mathscr C\}.$$
\begin{Remark}\label{44}
If $g(x)$ and $ h(x)$ are two polynomials with coefficients in a ring $A$, then $\deg(g(x)h(x))\leqslant  \deg(g(x))+\deg(h(x))$.  If $\operatorname{lc}(g)\operatorname{lc}(h)\neq 0$, then the equality holds, where ``$\operatorname{lc}$'' denotes the leading coefficient of each polynomial.
\end{Remark}
 \begin{Theorem}\label{17} Let $f_0$ be invertible in $\mathscr R$ and $\mathscr C$ be an $f(x)$-polycyclic code of length $n$ over $\mathscr R$. 
    \begin{enumerate}
        \item   The annihilator product is a non-degenerate $\mathscr R$-bilinear form.  
        \item The annihilator dual $\mathscr C^{\circ}$ is equal to $\operatorname{Ann}(\mathscr C)$.
        \item Let the leading coefficient of $g(x)\in\mathscr R_f$ be a unit element and $\mathscr C$ be a $f(x)$-polycyclic code generated by $g(x)$. If there exists a polynomial $h(x)\in \mathscr R_f$ satisfying $f(x)=g(x)h(x)$, then $\mathscr C^{\circ} = \langle h(x) \rangle$. Moreover, $|\mathscr C|\,|\mathscr C^{\circ}|=|\mathscr R|^n$.
    \end{enumerate}
 \end{Theorem}

 \begin{proof}\
     \begin{enumerate}
         \item It is straightforward that the given product is an $\mathscr R$-bilinear form. Moreover, it is easy to see that $x$ is invertible with inverse $m(x)=\sum_{i=0}^{n-2}-f_0^{-1}f_{i+1}x^i+f_0^{-1}x^{n-1}.$ To prove the non-degeneracy, suppose that  $g(x) = g_0 + g_1 x + \cdots + g_{n-1} x^{n-1} \in \mathscr R_f$ satisfies $\langle g(x) , h(x) \rangle_{f} = 0$ for all $h(x) \in \mathscr R_f.$ In the case $h=1$, we have $0=r(0)=g_0$. Hence, we can rewrite $g(x)$ as $(g_1+g_2x+\cdots+g_{n-1}x^{n-2})x$. In the case $h=x^{-1}$, we get $g_1=0$. Repeating this discussion, we conclude that $g_i = 0$ for all $i$, which implies $g(x) = 0$. Therefore, the annihilator product is non-degenerate.
         \item It is easy to see that $\operatorname{Ann}(\mathscr C)\subseteq \mathscr C^{\circ}.$ Conversely, let $g(x)\in \mathscr C^{\circ}$ and $h(x)$ be an arbitrary element of  $\mathscr C$. Let $r(x)=r_0 + r_1x + \cdots + r_{n-1}x^{n-1}$ denote the remainder of $g(x)h(x)$ modulo $f(x)$. We have  $\langle g(x), h(x)\rangle_f=0$, which implies  $r_0 = 0 $. Consequently, $(r_1+r_2x+\cdots+r_{n-1}x^{n-2})x=g(x)h(x)(\operatorname{mod}f(x)).$ According to the proof of the first statement, the inverse of $x$ is $m(x)$. Thus, $r_1+r_2x+\cdots+r_{n-1}x^{n-2}=g(x)h(x)m(x)(\operatorname{mod} f(x))$. Furthermore, $ghm(0)=0(\operatorname{mod} f(x))$, which implies $r_1=0$. By repeating the same reasoning, we have  $r_i = 0$ for all $i$, thus $gh = 0$.
         
         \item Since $f(x)=g(x)h(x)$, $\langle h(x) \rangle \subseteq \operatorname{Ann}(\langle g(x)\rangle )=\mathscr C^{\circ}$. Conversly, let $k(x)\in \operatorname{Ann}(\langle g(x)\rangle )$, then there exists a polynomial $t(x)$ such that $k(x)g(x) = t(x)f(x)$. If $\operatorname{lc}(g)\operatorname{lc}(h)=0$, then $\operatorname{lc}(h)=0$, which is a contradiction. Therefore, $\operatorname{lc}(g)\operatorname{lc}(h)\neq 0$, which implies that $\deg(g(x)h(x))=\deg(g(x))+\deg(h(x))$. Consequently, since $g(x)h(x)=f(x)$, we deduct $\operatorname{lc}(g)\operatorname{lc}(h)=\operatorname{lc}(f)=1$, which implies $\operatorname{lc}(h)$ is invertible. Therefore, we can apply the division algorithm to the monic polynomial $\operatorname{lc}(h)^{-1}h(x)$, that is, there exist unique elements $q(x), r(x)\in  \mathscr R_f$ such that $k(x) = q(x) \operatorname{lc}(h)^{-1}h(x)  + r(x)$ and $\deg(r(x)) < \deg (\operatorname{lc}(h)^{-1}h(x))$.  So, $k(x)g(x)=q(x)a^{-1}f(x)+r(x)g(x)$. This leads to $(t(x)-q(x)a^{-1})f(x)=r(x)g(x)$, and hence $r(x)g(x)\in \langle f(x)\rangle$ Assume that $r(x)$ is non-zero. As $\operatorname{lc}(g)$ is unit,  it follows that $ r(x)g(x)\neq 0$. Consequently, $0\neq r(x)g(x)\in \langle f(x)\rangle $, which contradicts the fact that $h(x)$ is the smallest polynomial satisfying $g(x)h(x) = f(x)$ and $\deg(r(x)) < \deg(\operatorname{lc}(h)^{-1}h(x))=\deg(h(x))$. Thus, $r(x)=0$, and hence $k(x)\in \langle h(x)\rangle$. As a result, $ \operatorname{Ann}(\langle g(x)\rangle )\subseteq \langle h(x) \rangle $. To complete the proof, note that $\langle f(x)/g(x)\rangle\oplus\langle g(x)\rangle=\mathscr R[x]/\langle f(x)\rangle$. Therefore, $|\mathscr C||\mathscr C^{\circ}|=|\mathscr R|^n$.
    \end{enumerate}
 \end{proof}
The set of vectors $ \{e_i : 1 \leqslant i \leqslant n\}$, where each $e_i$ has a value of $1$ in the $i$-th position and $0$ elsewhere, forms a basis for $\mathscr{R}^n$. The \emph{Gram matrix} $A$ associated with this set and the polynomial $f$ is a matrix with entries $a_{i,j}=\langle e_i, e_j\rangle_f$. 
For a linear code $\mathscr C$ over $\mathscr{R}$, we define the set $\mathscr C A := \{cA \mid c \in \mathscr C\}$. Similar to Lemma 4 and Corollary 1 in \cite{Fotue_2020}, we have the following lemma.
 \begin{Lemma}\label{18} Let $f_0\in \mathscr{R}$ be invertible element and $\mathscr C$ be an $f(x)$-polycyclic code of length $n$ over $\mathscr R$. 
    \begin{enumerate}
       \item $\mathscr C^\circ = (\mathscr C A)^\perp$.
       \item $(\mathscr C^{\circ})^{\circ} = \mathscr C$.
    \end{enumerate}
 \end{Lemma}
 \begin{Corollary}
 Let $f_0\in \mathscr{R}$  be an invertible element and $\mathscr C$ be a linear code of length $n$ over $\mathscr R$. Then, $\mathscr C$ is an $f(x)$-polycyclic code if and only if $\mathscr C^{\circ}$ is an $f(x)$-polycyclic code.
 \end{Corollary} 
 \begin{proof}
    It follows from the fact that $\operatorname{Ann}(\mathscr C)$ is an ideal of $\mathscr R_f$, and  $(\mathscr C^{\circ})^{\circ} = \mathscr C$.
 \end{proof}
  
 \begin{Theorem}\label{48} Let $f_0\in \mathscr{R}$  be an invertible element and  $\mathscr C$ be an $f(x)$-polycyclic code of length $n$ over $\mathscr R$ with the decomposition $\mathscr C \cong \oplus_{i=1}^k\mathscr C_i $.
    \begin{enumerate}
    \item  $\mathscr C^{\circ} \cong \oplus_{i=1}^k\mathscr C_i^{\circ}$.
    \item $\mathscr C$ is annihilator self-orthogonal (resp. self-dual, LCD, dual-containing) if and only if each $\mathscr C_i$ is annihilator self-orthogonal (resp. self-dual, LCD, dual-containing).
    \end{enumerate}
 \end{Theorem}
 
 \begin{proof}\
  \begin{enumerate}
    \item First, we claim that the annihilator product defined on each $\mathscr R_i[x]/\langle f_i(x)\rangle$ is non-degenerate. To prove it, note that since $\mathscr R \cong\oplus_{i=1}^k\mathscr R_{i}$, the group of units of the finite ring $\mathscr R$ decomposes as  $\mathscr R^* \cong\oplus_{i=1}^k\mathscr R^*_{i}$, see \cite[Chapter 1]{McDonald_1974}. Thus, since the constant part of the polynomial $f(x)$ is invertible, the constant part of each polynomial $f_i(x)$  is also invertible, and therefore, by Theorem \ref{17} the claim follows. According to the proof of the first statement in Proposition \ref{21}, the ring $\mathscr R[x]/\langle f(x)\rangle\cong \oplus_{i=1}^k\mathscr R_i[x]/\langle f_i(x)\rangle$ is isomorphic to $\oplus_{i=1}^k\mathscr R_i[x]/\langle f_i(x)\rangle e_{C_i}$, thus  $\mathscr C\cong \oplus_{i=1}^k\mathscr C_i\cong \oplus_{i=1}^k\mathscr C_ie_{C_i}$. Furthermore, since the entries of the Gram matrix $A$ are in $\mathscr R$, by Proposition \ref{13} yields $A = \oplus_{i=1}^k A_ie_{C_i} $, where $A_i$ is an $n \times n$ matrices over $\mathscr R$. Consequently, by Lemma \ref{18}, we have
      \begin{align*}
        \mathscr C^\circ = \left(\mathscr C A\right)^\perp & = \left(\bigg(\bigoplus_{i=1}^k\mathscr C_ie_{C_i}\bigg)\bigg(\bigoplus_{i=1}^k A_ie_{C_i}\bigg)\right)^\perp\\
                                                      &  = \left(\bigoplus_{i=1}^k\bigg(\mathscr C_i A_i\bigg) e_{C_i}\right)^\perp \\
                                                      &  = \bigoplus_{i=1}^k\bigg(\mathscr C_i A_i\bigg)^\perp e_{C_i}= \bigoplus_{i=1}^k\mathscr C_i^{\circ} e_{C_i}.
       \end{align*}
    \item If  $\mathscr C$ is annihilator self-orthogonal, then $\oplus_{i=1}^k\mathscr C_i e_{C_i}\subseteq \oplus_{i=1}^k\mathscr C_i^{\circ}e_{C_i}$. Multiplying the idempotent orthogonal $e_{C_i}$ on both sides implies $\mathscr C_i \subseteq \mathscr C_i^\circ$ for all $i$. The proof of the other statements follows a similar reasoning.
  \end{enumerate}
 \end{proof}

 \begin{Theorem}\label{47}
Let $f_0\in \mathscr{R}$  be an invertible element,  $\mathscr C$ is an $f(x)$-polycyclic code of length $n$ over $\mathscr R$ generated by a single polynomial $g(x) \in \mathscr R_f$, and the leading coefficient of $g(x)$ be a unit element in $\mathscr R$ and there exist $h(x) \in \mathscr R_f$ such that $f(x) = g(x)h(x)$. Then  
  \begin{enumerate}
         \item $\mathscr C$ is annihilator self-dual  if and only if $f(x)=ag^2(x)$ for some unit element $a$ in $\mathscr R$.
         \item $\mathscr C$ is annihilator self-orthogonal (dual-containing) if and only if $h(x) \mid g(x)$ ($g(x)\mid h(x)$).
        \item $\mathscr C$ is  annihilator LCD if and only if $\operatorname{gcd}(g(x), h(x)) = 1$.
  \end{enumerate}   
\end{Theorem}
\begin{proof}\
    \begin{enumerate}
        \item If $f(x)=ag^2(x)$, according to Theorem \ref{17}, we have $\mathscr C^{\circ}=\langle ag(x) \rangle=\langle g(x) \rangle=\mathscr C$. Conversely, let $\mathscr C$ be annihilator self-dual. Then, $g(x)\in \mathscr C$ implies $g(x)\in \mathscr C^\circ$. Thus, $g^2(x)=0 (\operatorname{mod}f(x))$, that means  there exists $l(x)$ in $\mathscr R[x]$  satisfying
        \begin{equation}\label{45}
            g^2(x)=l(x)f(x).
        \end{equation}   
        On the other side, by Theorem \ref{17}, we have $\mathscr C^{\circ}=\langle h(x) \rangle$. Since $\mathscr C^{\circ}=\mathscr C$, it follows that $\langle h(x) \rangle=\langle g(x) \rangle$. Thus, there exists $k(x)$ in $\mathscr R[x]$ such that $h(x)=k(x)g(x)$. Consequently,
    \begin{equation}\label{46}
        f(x)=k(x)g^2(x).
    \end{equation}  
   Now, since $\operatorname{lc}(g)\operatorname{lc}(g)\neq 0$,  $\operatorname{lc}(l)\operatorname{lc}(f)\neq 0$ and $\operatorname{lc}(k)\operatorname{lc}(g^2)\neq 0$, employing Remark \ref{44} and Eqs. \eqref{45} and \eqref{46}, we deduce
    $2\deg g=\deg l+\deg f$ and $\deg f=\deg k+2\deg g$, which imply $\deg l+\deg k=0$. Consequently, $\deg l=\deg k=0$. Thus, there exist elements $a, b\in R$ such that $l(x)=a$ , $k(x)=b$. Note that since leading coefficients of $f,g$ are units,  Eq.s \eqref{45} and \eqref{46} yield $a$ and $b$ are units.
    \item  Using Theorem \ref{17}, we obtain $\mathscr C=\langle g(x) \rangle$ and $\mathscr C^{\circ}=\langle h(x) \rangle$, which complete the proof. 
    \item Similar to the proof of Theorem 2.8 in \cite{Qi_2022}.
             
  \end{enumerate}
  \end{proof}
  
\begin{Remark}\label{56}
In Theorem 3 of \cite{Alahmadi_2016}, it is asserted ``A polycyclic code $\langle g \rangle$ in $\mathbb{F}[x]/\langle f \rangle$ is annihilator self-dual if and only if $f = g^2$''. The authors provided a proof for only one direction: ``If $f = g^2$, then the polycyclic code $\langle g \rangle$ in $\mathbb{F}[x]/\langle f \rangle$ is annihilator self-dual''. However, they omitted the proof for the converse statement. The proof of Theorem \ref{47} in this paper shows that the converse statement of their theorem is not true in general. Similarly, in Theorem 2.7 of \cite{Qi_2022}, the author attempted to prove Theorem 3 of \cite{Alahmadi_2016} for the case $\mathbb F$ is replaced with the ring $A=\mathbb F_q+\mathbb F_q$ with $u^2=0$. The author's proof relies on this statement: ``If $h(x)$ divides $g(x)$ and $g(x)$ divides $h(x)$, then $h(x)$ equals $g(x)$'', without considering that this divisibility relation cannot be true for elements in the ring of polynomials in general.
\end{Remark}

\begin{Corollary}\label{49}
    With the conditions in Theorem \ref{47}. Suppose that $\eta(g(x))=\big(g_1(x), g_2(x), \ldots, g_k(x)\big)$  and $\eta(h(x))=\big(h_1(x), h_2(x), \ldots, h_k(x)\big)$, where  the mapping $\eta$ defined in Equation \eqref{24}.
     \begin{enumerate}
         \item $\mathscr C$ is an annihilator self-dual code if and only if $f_i(x)=a_ig_i^2(x)$, for some unit elements $a_i$ in $\mathscr R_i$, for all $i\in\{ 1,\ldots, k\}$.
         \item $\mathscr C$ is an annihilator self-orthogonal (dual-containing) code if and only if $h_i(x) \mid g_i(x)$ ($g_i(x) \mid h_i(x)$) for all $i\in\{ 1,\ldots, k\}$.
         \item $\mathscr C$ is an annihilator LCD code if and only if $\operatorname{gcd}(g_i(x), h_i(x)) = 1$ for all $i\in\{ 1,\ldots, k\}$.
  \end{enumerate}   
\end{Corollary}

\begin{proof}
   We will proof the first statement, and the other ones follow similarly. Note that since $\mathscr R \cong\oplus_{i=1}^k\mathscr R_{i}$, the group of units of the  ring $\mathscr R$ decomposes as  $\mathscr R^* \cong\oplus_{i=1}^k\mathscr R^*_{i}$  \cite[Chapter 1]{McDonald_1974}. Thus, if the leading coefficient of $g(x)$ is a unit element in $\mathscr R$, then the leading coefficient of $g_i(x)$ is a unit element in $\mathscr R_i$ for all $i\in\{ 1,\ldots, k\}$. Therefore, by Theorem \ref{47}, each $\mathscr C_i$  is an annihilator self-dual code if and only if $f_i(x)=a_ig_i^2(x)$ for some unit element $a_i$  in $\mathscr R_i$. Then, the proof follows by Theorem \ref{48}.
\end{proof}

\begin{Example}\label{sh1}
    Consider the ring $\mathscr R = R[x_1, x_2, x_3]/\langle x_1^2-1, x_2^2-1, x_3^2-1 \rangle$, where $R$ is a chain ring with the residue field $ \mathbb F_q$. With a similar approach as in \cite{Ashraf_2022}, any element $r \in \mathscr R$ of the form $r = a_1 + a_2x_1 + a_3x_2 + a_4x_3 + a_5x_1x_2 + a_6x_2x_3 + a_7x_1x_3 + a_8x_1x_2x_3$ can be uniquely expressed as $r_1e_1+r_2e_2+r_3e_3+r_4e_4+r_5e_5+r_6e_6+r_7e_7+r_7e_8$, where $r_i\in \mathscr R$ and
   \begin{align*}
e_1 &= t(1 + x_1 + x_2 + x_3 + x_1x_2 + x_2x_3 + x_1x_3 + x_1x_2x_3),\\
e_2 &= t(1 + x_1 + x_2 - x_3 + x_1x_2 - x_2x_3 - x_1x_3 -x_1x_2x_3 ),\\
e_3 &= t(1 + x_1 - x_2 + x_3 - x_1x_2 - x_2x_3 + x_1x_3 - x_1x_2x_3),\\
e_4 &= t(1 - x_1 + x_2 + x_3 - x_1x_2 + x_2x_3 - x_1x_3 - x_1x_2x_3),\\
e_5 &= t(1 + x_1 - x_2 - x_3 - x_1x_2 + x_2x_3 - x_1x_3 + x_1x_2x_3),\\
e_6 &= t(1 - x_1 - x_2 + x_3 + x_1x_2 - x_2x_3 - x_1x_3 + x_1x_2x_3),\\
e_7 &= t(1 - x_1 + x_2 - x_3 - x_1x_2 - x_2x_3 + x_1x_3 + x_1x_2x_3),\\
e_8 &= t(1 - x_1 - x_2 - x_3 + x_1x_2 + x_2x_3 + x_1x_3 - x_1x_2x_3),
\end{align*}
where $t\in \mathbb F_q$ and $8t=1 \mod p$.
As a special case, let $R=\mathbb Z_9$ and $t=2$. 
\begin{enumerate}
 \item Suppose that $f_0=2x_1-2x_2-2x_3+x_1x_2x_3$, $f_1=x_1+2x_3$,$f_2=1-2x_1+2x_2+2x_3+x_1x_2+x_2x_3$, $g_0=x_1-2x_3+x_1x_2$, $g_1=-2x_1x_2-2x_2x_3-2x_1x_3$ and $g_2=1+x_1+x_2+x_3+x_1x_2+x_2x_3+x_1x_3+x_1x_2x_3$ are elements in $\mathscr R$. Let $\mathscr C$ be a code in $\mathscr R[x]/\langle f(x)\rangle$ generated by $g(x)=g_2x^2+g_1x+g_0$, where $f(x)=x^3-f_2x^2-f_1x-f_0$.  It is easy to see that
\begin{align*}
   f_0=& 2e_1+e_2+e_5+2e_8, & g_0=& 2e_1+2e_2+e_5+2e_8,\\
   f_1=&2e_4+e_5+2e_7+e_8, & g_1=&e_1+2e_2+2e_3+e_4+e_6+2e_7,\\
   f_2=&2e_4+2e_5+e_6+e_7+e_8,  & g_2=& 2e_1+2e_2+2e_3++2e_4+2e_5+2e_6+2e_7+2e_8.
\end{align*}
Therefore, 
\begin{align*}
    f_1(x) &= x^3 + 1, & f_2(x) &= x^3 + 2, & f_3(x) &= x^3 + x^2, & f_4(x) &= x^3 + x^2 + x, \\
    f_5(x) &= x^3 + x^2 +2x+ 2, & f_6(x) &= x^3 + 2x^2, & f_7(x) &= x^3 + 2x^2 + x , & f_8(x) &= x^3 + 2x^2 + 2x + 1,
\end{align*}
 and 
\begin{align*}
    g_1(x) &=  2x^2+x+2, & g_2(x) &=  2x^2+2x+2, & g_3(x) &= 2x^2+2x, & g_4(x) &= 2x^2 + x, \\
    g_5(x) &=  2x^2 +1, & g_6(x) &=  2x^2+x, & g_7(x) &=  2x^2 + 2x , & g_8(x) &= 2x^2 +1.
\end{align*}
One can easily see that  $g^2_i(x)=f_i(x)$ for all $i$. Using Corollary \ref{49},  $\mathscr C$ is an annihilator self-dual code in $\mathscr R[x]/\langle f(x)\rangle$. 

\item Suppose now that $f_0=-x_3-x_1x_2+x_2x_3-x_1x_2x_3$, $f_1=2x_1+x_1x_2+2x_2x_3+2x_1x_3+2x_1x_2x_3$, $f_2=1+2x_2-2x_3+2x_1x_2-2x_1x_3$,    $g_0=2x_1+2x_2+x_3+x_2x_3+x_1x_3$, $g_1=2+x_1+2x_2+x_3+2x_1x_2+2x_1x_3$ and $g_2=2+x_1+2x_2+x_3+x_1x_2+x_2x_3+x_1x_3+x_1x_2x_3$ are elements in $\mathscr R$. Let $\mathscr C$ be a code in $\mathscr R[x]/\langle f(x)\rangle$ generated by $g(x)=g_2x^2+g_1x+g_0$, where $f(x)=x^3-f_2x^2-f_1x-f_0$.  It is easy to see that
\begin{align*}
   f_0=& 2e_1+2e_2+2e_3+2e_4+2e_5+2e_6+e_7+2e_8, & g_0=& 2e_1+2e_2+e_5+e_8, \\
   f_1=&e_1+2e_4+2e_6+e_8, & g_1=&2e_1+e_4+e_6+2e_7, \\
   f_2=&2e_2+e_3+e_4+e_8, &  g_2=&2e_1+2e_2+2e_3++2e_4+2e_5+2e_6+2e_7+2e_8.
\end{align*}
Therefore, 
\begin{align*}
    f_1(x) &= x^3 + 2x + 1, & f_2(x) &= x^3 + x^2 + 1, & f_3(x) &= x^3 + 2x^2 + 1, & f_4(x) &= x^3 + 2x^2 + x + 1, \\
    f_5(x) &= x^3 + 1, & f_6(x) &= x^3 + x + 1, & f_7(x) &= x^3 + 2 , & f_8(x) &= x^3 + 2x^2 + 2x + 1,
\end{align*}
 and 
\begin{align*}
    g_1(x) &=  x^2 + 2x + 2, & g_2(x) &=  x^2 + 2x + 2, & g_3(x) &= x^2 + x + 2, & g_4(x) &= x^2 + x + 2, \\
    g_5(x) &=  x^2 + 2x + 2, & g_6(x) &=  x^2 + x + 2, & g_7(x) &= x^2 + x + 2 , & g_8(x) &= x^2 + 2x + 2.
\end{align*}
Assume that $h_i(x)=x+1$.  The greatest common divisor of \( g_i(x) \) and \( h_i(x) \) is $1$ for all $i$. Moreover, $f_i(x)=g_i(x)h_i(x)$.  By Corollary \ref{49}, the code $\mathscr C$ generated by $g(x)$ is an annihilator LCD code in $\mathscr R[x]/\langle f(x)\rangle$.
\end{enumerate}
\end{Example}
\section{Annihilator CSS construction}\label{76}
{The Euclidean CSS construction is applied to classical linear codes that preserve Euclidean duals to get quantum codes, see for example \cite{Nadella_2012}. In this section, we aim to introduce the annihilator CSS construction for polycyclic codes based on the annihilator dual.}

Let $R$ be a finite commutative Frobenius ring with $\mathfrak{q}$ elements. Let $f(x) = x^n - \sum_{r=0}^{n-1} f_{r}x^{r}\in R[x]$ and $C$ be an $f(x)$-polycyclic code. According to Section 6, the annihilator product $\langle g(x), h(x) \rangle_f$ is defined as $r(0)$, where $r(x)$ is the reminder of $gh(x)$ divided by $f$. The annihilator dual with respect to this product is denoted by $C^{\circ}.$ The Gram matrix $A$ associated with the polynomial $f$ is a $n\times n$ matrix with entries $a_{i,j}=\langle x^{i-1}, x^{j-1}\rangle_f$. For example, if  $f(x) = x^5 - f_4x^4 - f_3x^3 - f_2x^2 - f_1x - f_0$, then
\begin{equation}\label{123}
A = 
\begin{pmatrix}
    1 & 0 & 0 & 0 & 0 \\[4pt]
    0 & 0 & 0 & 0 & f_0 \\[4pt]
    0 & 0 & 0 & f_0 & f_0f_4 \\[4pt]
    0 & 0 & f_0 & f_0f_4 & f_0(f_4^2 + f_3) \\[4pt]
    0 & f_0 & f_0f_4 & f_0(f_4^2 + f_3) & f_0\left(f_4^3 + 2f_4f_3+f_2\right)
\end{pmatrix}.
\end{equation}
The Gram matrix $A$ is symmetric, and if  $f_0$ is invertible, then $A$ is also invertible. The following lemma provides a recursive formula for calculating the matrix $A$. 
\begin{Lemma}
 Let $f(x) = x^n - \sum_{r=0}^{n-1} f_{r}x^{r}\in R[x]$ and $A=(a_{i,j})$ be an $n\times n$ Gram matrix associated with $f$. Then  $a_{i,j} \;=\; c_{\,i+j-2}$, where
 \begin{equation*}
       c_k=
       \begin{cases}
         1, & k=0,\\
         0, & 0< k<n,\\
         \sum_{r=0}^{n-1} f_r\,c_{k-n+r}, & k\ge n.
       \end{cases}  
 \end{equation*}
\end{Lemma}
\begin{proof}
To compute each $a_{i,j}$, it suffices to find the constant term of the associated polynomial
\begin{equation*}
x^{\,i+j-2}\bigr\rvert_{x^n = f_{n-1}x^{n-1} + \cdots + f_1 x + f_0}.
\end{equation*}
Let us denote the constant term of $x^k\vert_{x^n}$ by $c_k$. Therefore, $A$ is the following Hankel matrix\footnote{A matrix is called Hankel if each ascending skew-diagonal from left to right is constant.}
\begin{equation*}
A=\begin{pmatrix}
c_0 & c_1 & c_2 & \cdots & c_{n-2} & c_{n-1} \\
c_1 & c_2 & c_3 & \ldots & c_{n-1} & c_n \\
c_2 & c_3 & c_4 & \ldots & c_n & c_{n+1} \\
\vdots & \vdots  & \vdots  & \ddots & \vdots  & \vdots \\
c_{n-2} & c_{n-1} & c_n & \ldots & c_{2n-4} & c_{2n-3} \\
c_{n-1} & c_n & c_{n+1} & \ldots & c_{2n-3} & c_{2n-2}
\end{pmatrix}.
\end{equation*}
It is clear that $c_k=1$ if $k=0$, and $c_k=0$ if $0 < k< n$. For $k\geqslant n$, we have $x^k=x^{k-n} x^n= x^{k-n}\sum_{r=0}^{n-1} f_r\,x^r=\sum_{r=0}^{n-1} f_r\, x^{(k-n)+r}$. Therefore, 
\begin{equation*}
c_k =\left(\text{constant term of }x^k\bigr\rvert_{x^n}\right)=
\sum_{r=0}^{n-1} f_r \left( \left. \text{constant term of } x^{(k-n)+r} \right|_{x^n} \right)= \sum_{r=0}^{n-1} f_r\,c_{\,k-n+r}.
\end{equation*}
\end{proof}

\begin{Lemma}\label{khk4}
 If $R$ is a commutative ring and $f_0$ is invertible, then $C^{\circ}=(C A)^\perp= C^\perp A^{-1}$, where $C A=\{cA\,|\,c\in C\}$.
\end{Lemma}
\begin{proof}
The first equality can be proved similarly to Lemma 4 in \cite{Fotue_2020}. Consider the vectors $c, v \in {R}^n$. Since $A$ is symmetric, we have
$\langle cA,  v\rangle=cAv^\text{tr}=c(vA)^\text{tr}=\langle c,  vA\rangle$, where $\langle \,,\,\rangle$ denotes the Euclidean inner product. So, $(C A)^\perp=\bigl\{v \in {R}^n \mid \langle cA, v\rangle = 0,\,\forall c\in C\bigr\}=\bigl\{v\in {R}^n \mid \langle c , vA\rangle = 0,\,\forall c\in C\bigr\}$. Thus, $v \in (C A)^\perp$ if and only if $vA \in C^\perp$. Since $A$ is invertible, this is equivalent to $v \in (C A)^\perp$ if and only if $v \in C^\perp A^{-1}$. Therefore, $(C A)^\perp= C^\perp A^{-1}$.
\end{proof}

\begin{Remark}
Suppose that $f(x)=x^n - f_0\in R[x]$ with $f_0$ invertible. Then $A$ and  its inverse are the monomial matrices
\begin{equation*}
A=  
\begin{pmatrix}
1 & 0 & \ldots & 0 & 0 \\
0 & 0 & \ldots & 0 & f_0 \\
0 & 0 & \ldots & f_0 & 0 \\
\vdots & \vdots & \iddots & \vdots & \vdots \\
0 & f_0 & \ldots & 0&0 
\end{pmatrix},\qquad 
A^{-1}=  
\begin{pmatrix}
1 & 0 & \ldots & 0 & 0 \\
0 & 0 & \ldots & 0 & f_0^{-1} \\
0 & 0 & \ldots & f_0^{-1} & 0 \\
\vdots & \vdots & \iddots & \vdots & \vdots \\
0 & f_0^{-1} & \ldots & 0&0 
\end{pmatrix}.
\end{equation*}
In this case, since $C^{\circ}=C^\perp A^{-1}$, the dual codes $C^\circ$ and $C^\perp$ are monomially equivalent, which ensures they have the same minimum Hamming distance. However, in general, $A$ may not preserve the minimum Hamming distance because it may not preserve the Hamming weight. 
\end{Remark}

Below is a well-known theorem for constructing quantum codes from classical linear codes that preserve the Euclidean dual. Note that $[[n, k, d]]_{\mathfrak{q}}$ will denote a $R$-ary quantum code of length $n$, with $\mathfrak{q}^k$ codewords, and the minimum distance $d$.
\begin{Theorem}[Euclidean CSS construction \cite{Nadella_2012}]\label{khk3}
Let $C_1$ and $C_2$ be two linear codes over the finite commutative Frobenius ring $R$ with parameters $[n,k_1,d_1]_R$ and $[n,k_2,d_2]_R$, respectively. If $C_2^{\perp} \subseteq C_1$, then there exists a quantum stabilizer code with parameters $[[n, k_1+k_2-n,d]]_{\mathfrak{q}}$, where $d = \min\{\operatorname{wt}(c)\mid c \in (C_1 \setminus C_2^{\perp}) \cup (C_2 \setminus C_1^{\perp})\}$ is the minimum Hamming distance.  
\end{Theorem}
 Now, we are ready to introduce the CSS construction for polycyclic codes that preserve annihilator duals.
\begin{Theorem}[Annihilator CSS construction]\label{sh6}
Suppose that $R$ is a Frobenius ring and $f_0\in R$ is invertible. Let $f(x)=g_1(x)h_1(x)=g_2(x)h_2(x)$, and let $C_1=\langle g_1(x)\rangle$ and $C_2=\langle g_2(x)\rangle$ be two $f(x)$-polycyclic codes over $R$ with parameters $[n, k_1, d_1]_R$ and $[n, k_2, d_2]_R$, respectively. If  $C^\circ_2 \subseteq C_1$ then there exists a quantum stabilizer code with parameters $[[n, k_1+k_2-n, d]]_{\mathfrak{q}}$, where $d = \min\{\operatorname{wt}(c)\mid c \in (C_1 \setminus C_2^{\circ}) \cup (C_2A \setminus C_1^{\circ}A)\}$ is the minimum Hamming distance.
\end{Theorem}

\begin{proof}
Since the $n\times n$ Gram matrix $A$ is invertible over $R$, then $xA = 0$ leads to $x=0$ for any $x\in R^n$ (similar to Corollary 2.4 in \cite{Liu_2015}). So the transformation $x\mapsto xA$ is a bijective over $R^n$, which implies that the dimension of the code $C_2 A$ equals to $ k_2$. So, if we replace $C_2$ by $C_2A$ in the Theorem above, then there exists a quantum stabilizer code with parameters $[[n, k_1+k_2-n, d]]_{\mathfrak{q}}$, where  $d = \min\{\operatorname{wt}(c)\mid c \in (C_1 \setminus (C_2A)^{\perp}) \cup (C_2A \setminus C_1^{\perp})\}$. Applying Lemma \ref{khk4}, we obtain $d = \min\{\operatorname{wt}(c)\mid c \in (C_1 \setminus C_2^{\circ}) \cup (C_2A \setminus C_1^{\circ}A)\}$. 
\end{proof}

\begin{Corollary}\label{71}
Suppose that $R$ is a finite Frobenius ring and $f_0\in R$ is invertible. Let $f(x)=g(x)h(x)$, and let $C=\langle g(x)\rangle$ be an  $f(x)$-polycyclic code over $R$ with parameters $[n,k,d_1]_R$. If $C^\circ \subseteq C$, then there exists an $[[n, 2k-n, d]]_{\mathfrak{q}}$ stabilizer code, where $ d=\min\{\operatorname{wt}(c) \mid c \in (C \setminus C^{\circ}) \cup (CA \setminus C^{\circ}A)\}$. 
\end{Corollary}

\begin{Example}
 Consider the serial ring $\mathscr R = \mathbb F_3[x_1, x_2]/\langle x_1^2-x_1, x_2^2-x_2 \rangle$, which is isomorphic to $\mathbb F_3+u\mathbb F_3+v\mathbb F_3+uv\mathbb F_3$ with $u^2=u$, $v^2=v$, and $uv=vu$. According to Proposition \ref{32}, $\mathscr R$ is a Frobenius ring.  Let $f(x)=x^5+2x^3+x^2+2=(x+1)^4(x+2)\in \mathscr R[x]$. The $f(x)$-polycyclic code $\mathscr C$ generated by $g(x)=x+1$ is a free $\mathscr R$-module of dimension $k=\deg(f)-\deg(g)$ with the generating set $\{g(x), xg(x),\ldots, x^{k-1}g(x)\}$. Consider the Gray map $\varphi : \mathscr R \rightarrow \mathbb F_3^4$ defined as $\varphi(a +ub+vc+uvd) =(a,a +b,a +c,a +b+c+d)$ in \cite{Ashraf_2016}. Using equation \eqref{123}, one can compute the Gram matrix $A$. Aplying Theorem \ref{47}, $\mathscr C$ is annihilator dual-containing. Using Magma~\cite{Magma}, we compute that $d=2$. Therefore, we obtain a $[[5, 3, 2]]_{3^4}$ quantum stabilizer code.
\end{Example}
\section{Conclusion and future work}
In this paper, we have investigated the structure of $f(x)$-polycyclic codes over a specific category of serial rings, defined as $\mathscr R=R[x_1,\ldots, x_s]/\langle t_1(x_1),\ldots, t_s(x_s) \rangle$, where $R$ is a chain ring and each $t_i(x_i)$ in $R[x_i]$ is a monic square-free polynomial. We have established a unified framework in the literature of constacyclic and polycyclic codes over various rings, all of which are examples of the ring $\mathscr R$. We have defined the new concepts QsDP codes and $f(x)$-polycyclic-QsDP codes, and have proved that $f(x)$-polycyclic codes over $\mathscr R$ are isomorphic to $f(x)$-polycyclic-QsDP codes. Additionally,   we have established a polynomial structure for the Euclidean dual of $f(x)$-polycyclic codes over $\mathscr R$. The annihilator dual is an alternative duality concept ensuring that the dual of $f(x)$-polycyclic codes have an $f(x)$- polycyclic structure, a feature that fails in the Euclidean dual. We have studied some properties of the annihilator dual of $f(x)$-polycyclic code over the ring $\mathscr R$ and have provided necessary and sufficient conditions for the existence of annihilator self-dual, annihilator self-orthogonal, annihilator LCD, and annihilator dual-containing polycyclic codes over $\mathscr R$. Finally, we have provided an annihilator CSS construction to derive quantum codes from annihilator dual-preserving $f(x)$-polycyclic codes over Frobenius rings. 
The use of a general polynomial other than $x^n-1$ or $x^n-\lambda$ (cyclic and constacyclic case) provides the chance to explore a broader range of factors, and the annihilator duality provides that the dual code is always in the same ambient space and it is easy to compute, hence the annihilator CSS construction provides a broader class of codes where to find quantum codes.
As a future work, it will be worth studying the annihilator duality and checking if it yields new quantum codes. Also in \cite{Bag_2025}, the authors have recently studied quasi-polycyclic codes and their application for constructing quantum error-correcting codes based on the CSS construction for Euclidean and symplectic duality. Extending these results to the annihilator dual setting may yield new classes of quantum codes.


\begin{thebibliography}{99}
\bibitem{Alahmadi_2016} Alahmadi, A.,  Dougherty, S. T., Leroy, A., \& Sol\'e, P. (2016). On the duality and the direction of polycyclic codes. Adv. Math. Commun., 10(4), 921-929.

\bibitem{Ali_2024} Ali, S., Alali, A. S., Wong, K. B., Oztas, E. S., \& Sharma, P. (2024). Cyclic codes over non-chain ring $ \mathcal{R}(\alpha_1, \alpha_2, \ldots, \alpha_s) $ and their applications to quantum and DNA codes. AIMS Mathematics, 9(3), 7396-7413.

\bibitem{Alkenani_2020} Alkenani, Ahmad N., Mohammad Ashraf, \& Ghulam Mohammad. (2020). Quantum Codes from Constacyclic Codes over the Ring $\mathbb F_q[u_1, u_2]/\langle u_{1}^2 - u_1, u_{2}^2 - u_2, u_1u_2 - u_2u_1 \rangle$. Mathematics 8(5) 781. 

\bibitem{Ashraf_2022} Ashraf, M., Khan, N., \& Mohammad, G. (2022). Quantum codes from cyclic codes over the mixed alphabet structure. Quantum Information Processing, 21(5), 180.

\bibitem{Ashraf_2016} Ashraf, M., \& Mohammad, G. (2016). Quantum codes from cyclic codes over $ F_q+ uF_q+ vF_q+ uvF_q $. Quantum Information Processing, 15(10).

\bibitem{Bag_2020} Bag, T., Dinh, H. Q., Upadhyay, A. K., Bandi, R., \& Yamaka, W. (2020). Quantum codes from skew constacyclic codes over the ring $\mathbb Fq [u, v]/\langle u^2- 1, v^2- 1, uv-vu\rangle$. Discrete Mathematics, 343(3), 111737.

\bibitem{Bag_2025} Bag, T., \& Panario, D. (2025). Quasi-polycyclic and skew quasi-polycyclic codes over $\mathbb F_q$. Finite Fields and Their Applications, 101, 102536.

\bibitem{Bajalan_2022} Bajalan, M., Mart\'{i}nez-Moro, E., \& Szabo, S. (2022). A transform approach to polycyclic and serial codes over rings. Finite Fields and Their Applications, 80, 102014. 

\bibitem{Bajalan_2024} Bajalan, M., Landjev, I., Mart\'{i}nez-Moro, E., \& Szabo, S. (2023). $(\sigma,\delta) $-polycyclic codes in Ore extensions over rings. arXiv preprint: arXiv:2312.07193.

\bibitem{Bhardwaj_2023} Bhardwaj, S., Goyal, M., \& Raka, M. (2023). New quantum codes from constacyclic codes over a general non-chain ring. Discrete Mathematics, Algorithms and Applications, Online Ready.

\bibitem{Magma} Wieb Bosma, John Cannon, and Catherine Playoust, The Magma algebra system. I. The user language, J. Symbolic Comput., 24 (1997), 235--265. 

\bibitem{Calderbank_1998} Calderbank, A. R., Rains, E. M., Shor, P. M., \& Sloane, N. J. (1998). Quantum error correction via codes over GF (4). IEEE Transactions on Information Theory, 44(4), 1369-1387.

\bibitem{Cazaran_1999} Cazaran J. \& Kelarev A. V. (1999) On finite principal ideal rings, Acta Math. Univ. Comenianae, 68(1), 77-84.

\bibitem{Dinh_2020} Dinh, H. Q., Bag, T., Pathak, S., Upadhyay, A. K., \& Chinnakum, W. (2020). Quantum codes obtained from constacyclic codes over a family of finite rings $\mathbb F_p [u_1, u_2, \ldots, u s]$. IEEE Access, 8, 194082-194091.

\bibitem{Dinh_2004} Dinh, H. Q., \& Lopez-Permouth, S. R. (2004). Cyclic and negacyclic codes over finite chain rings. IEEE Transactions on Information Theory, 50(8), 1728-1744.

\bibitem{Dinh_2017} Dinh, H. Q., Nguyen, B. T., \& Sriboonchitta, S. (2017). Constacyclic codes over finite commutative semi-simple rings. Finite Fields and Their Applications, 45, 1-18. 

\bibitem{Dougherty_2017} Dougherty, S. T. (2017). Algebraic coding theory over finite commutative rings. In Springer Briefs in Mathematics. Springer International Publishing. 

\bibitem{Fotue_2020} Fotue-Tabue, A., Mart\'{i}nez-Moro, E., \& Blackford, J. T. (2020). On polycyclic codes over a finite chain ring. Advances in Mathematics of Communications, 14(3).

\bibitem{Goyal_2020} Goyal, M., \& Raka, M. (2020). Polyadic constacyclic codes over a non-chain ring $\mathbb F_q[u,v]/ \langle f(u),g(v),uv -vu\rangle$. Journal of Applied Mathematics and Computing, 62(1), 425-447.

\bibitem{Goyal_2021} Goyal, M., \& Raka, M. (2021). Polyadic cyclic codes over a non-chain ring $\mathbb F_q[u,v]/ \langle f(u),g(v),uv -vu\rangle$. Journal of Computer and Communications, 9(5), 36-57.

\bibitem{Guenda_2014} Guenda, K., \& Gulliver, T. A. (2014). Quantum codes over rings. International Journal of Quantum Information, 12(04), 1450020.

\bibitem{Guneri_2016} G\"uneri, C., \& \"Ozkaya, B. (2016). Multidimensional quasi-cyclic and convolutional codes. IEEE Transactions on Information Theory, 62(12), 6772-6785.

\bibitem{Guneri_2012} G\"uneri, C., \& \"Ozbudak, F. (2012). A relation between quasi-cyclic codes and 2-D cyclic codes. Finite Fields and Their Applications, 18(1), 123-132. 

\bibitem{Hassan_2022} Hassan, O. A., Brahim, I., \& El Mahdi, M. (2022). A class of 1-generator right quasi-polycyclic Codes. 2022 2nd International Conference on Innovative Research in Applied Science, Engineering and Technology (IRASET). 

\bibitem{Islam_2021} Islam, H., Mart\'inez-Moro, E., \& Prakash, O. (2021). Cyclic codes over a non-chain ring $R_{e,q}$ and their application to LCD codes. Discrete Mathematics, 344(10), 112545.

\bibitem{Karthick_2023} Karthick, G. (2023). Polycyclic codes over $R$. Communications in Combinatorics and Optimization,  in press.
\bibitem{Liu_2015} Liu, X., \& Liu, H. (2015). LCD codes over finite chain rings. Finite Fields and Their Applications, 34, 1-19.

\bibitem{Permouth_2009} L\'opez-Permouth, S. R., Parra-Avila, B. R., \& Szabo, S. (2009). Dual generalizations of the concept of cyclicity of codes. Adv. Math. Commun., 3(3), 227-234.

\bibitem{Moro_2006} Mart\'{i}nez-Moro, E., \& R\'{u}a, I. F. (2006). Multivariable codes over finite chain rings: serial codes. SIAM Journal on Discrete Mathematics, 20(4), 947-959.

\bibitem{moro_2007} Mart\'{i}nez-Moro, E., \& R\'{u}a, I. F. (2007). On repeated-root multivariable codes over a finite chain ring. Designs, Codes and Cryptography, 45, 219-227.

\bibitem{Moro_2018b} Mart\'{i}nez, E., Pi\~nera-Nicol\'as, A., \& R\'{u}a, I. F. (2018). Multivariable codes in principal ideal polynomial quotient rings with applications to additive modular bivariate codes over $\mathbb F_4$. Journal of Pure and Applied Algebra, 222(2), 359-367.

\bibitem{Moro_2018} Mart\'{i}nez-Moro, E., Pi\~{n}era-Nicol\'as, A., \&  R\'{u}a, I. F. (2018). Codes over affine algebras with a finite commutative chain coefficient ring. Finite Fields and Their Applications, 49, 94-107.

\bibitem{McDonald_1974} McDonald, B. R. (1974). Finite rings with identity. In Pure and Applied Mathematics (Vol. 28). Marcel Dekker.

\bibitem{Nadella_2012} Nadella, S. (2012). Stabilizer Codes over Frobenius Rings. PhD Dissertation. Texas A\&M University.

\bibitem{Patel_2022} Patel, S., \& Prakash, O. (2022, June). Quantum codes construction from skew polycyclic codes. In 2022 IEEE International Symposium on Information Theory (ISIT) (pp. 1070-1075). IEEE.

\bibitem{Puninski_2001} G. Puninski, Serial Rings, Kluwer Academic Publishers, Dordrecht, The Netherlands, (2001).

\bibitem{Qi_2022} Qi, W. (2022). On the polycyclic codes over $ \mathbb{F}_q+u\mathbb{F}_q $. Advances in Mathematics of Communications, in press. 

\bibitem{Sarvepalli_2008} Sarvepalli, P. K. (2008). Quantum stabilizer codes and beyond. PhD Dissertation. Texas A\& M University.

\bibitem{Shi_2017} Shi, M., Alahmadi, A., \& Sole, P. (2017). Codes and rings: theory and practice. Academic Press.


\bibitem{Wu_2022} Wu, R., Shi, M., \& Sole, P. (2022). On the structure of 1-generator quasi-polycyclic codes over finite chain rings. Journal of Applied Mathematics and Computing, 1-13.



\end{thebibliography}
\end{document}